\documentclass[a4paper]{amsart}\usepackage[T1]{fontenc}
\usepackage[utf8]{inputenc}
\usepackage[english]{babel}
\usepackage{amsmath}
\usepackage{amssymb}
\usepackage{amsfonts}
\usepackage{amsrefs}
\usepackage{amsaddr}
\usepackage{MnSymbol}
\usepackage{stmaryrd}
\usepackage{mathrsfs} 
\usepackage{natbib}
\usepackage{graphicx,graphics}
\usepackage{slashed}
\usepackage{esint}
\usepackage{color}
\usepackage{multirow}
\usepackage[colorlinks]{hyperref}

\newtheorem{theorem}{Theorem}

\newenvironment{theorem*}[1][Theorem]{\begin{trivlist}\itshape
  \item[\hskip \labelsep {\bfseries #1}]}{\normalfont\end{trivlist}}
\newenvironment{remark*}[1][Remark]{\begin{trivlist}
  \item[\hskip \labelsep {\bfseries #1}]}{\hfill$\square$\end{trivlist}}
\newenvironment{example*}[1][Example]{\begin{trivlist}
  \item[\hskip \labelsep {\bfseries #1}]}{\hfill$\square$\end{trivlist}}

\allowdisplaybreaks

\DeclareMathSymbol{\eth} {\mathord}{AMSb}{"67}

\renewcommand{\phi}{\varphi}

\renewcommand{\rho}{\varrho}

\renewcommand{\theta}{\vartheta}

\renewcommand{\d}{\partial}

\newcommand{\vol}{\ensuremath{\mathrm{vol}}}

\newcommand{\ex}{\exists}
\newcommand{\fa}{\forall}
\newcommand{\lnorm}{\left\lVert}
\newcommand{\rnorm}{\right\lVert}
\newcommand{\lbetr}{\left\lvert}
\newcommand{\rbetr}{\right\lvert}
\newcommand{\norm}[1]{\lnorm {#1}\rnorm}

\newcommand{\abs}[1]{\lbetr {#1}\rbetr}

\renewcommand{\l}{\ensuremath{\left}}
\renewcommand{\r}{\ensuremath{\right}}
\newcommand{\sse}{\ensuremath{\subseteq}}

\newcommand{\tr}{\ensuremath{\opn{tr}}}

\newcommand{\opn}{\operatorname}

\let\stdparagraph\paragraph
\renewcommand\paragraph{\vspace*{1em}\stdparagraph}

\renewcommand{\cite}[1]{{[\cites{#1}}}

\newcommand{\Gf}{\ensuremath{\mathfrak{G}}}
\newcommand{\gf}{\ensuremath{\mathfrak{g}}}

\newcommand{\Ap}{\ensuremath{\mathcal{A}}}
\newcommand{\Dp}{\ensuremath{\mathcal{D}}}
\newcommand{\Np}{\ensuremath{\mathcal{N}}}
\newcommand{\Sp}{\ensuremath{\mathcal{S}}}
\newcommand{\Lp}{\ensuremath{\mathcal{L}}}

\newcommand{\rn}{\ensuremath{\mathbb{R}}}
\newcommand{\cn}{\ensuremath{\mathbb{C}}}
\newcommand{\nn}{\ensuremath{\mathbb{N}}}
\newcommand{\zn}{\ensuremath{\mathbb{Z}}}

\newcommand{\diag}{\ensuremath{\mathrm{diag}}}
\newcommand{\sgn}{\ensuremath{\mathrm{sgn}}}
\newcommand{\spt}{\ensuremath{\mathrm{spt}}}

\newcommand{\eref}[1]{\hyperref[#1]{Equation~\eqref{#1}}}

\begin{document}
\title[Zeta-regularized Lattice Field Theory]{Zeta-regularized Lattice Field Theory with Lorentzian background metrics}

\author{Tobias Hartung}
\address{Department of Mathematical Sciences, University of Bath, 4 West, Claverton Down, Bath, BA2 7AY, United Kingdom\\and\\Computation-based  Science  and  Technology  Research  Center, The  Cyprus  Institute,  20  Kavafi  Street,  2121  Nicosia, Cyprus}

\author{Karl Jansen}
\address{NIC, DESY Zeuthen, Platanenallee 6, 15738 Zeuthen, Germany}

\author{Chiara Sarti}
\address{Computer Laboratory, University of Cambridge, William Gates Building, 15 JJ Thomson Avenue, Cambridge, CB3 0FD, United Kingdom}


\maketitle

\begin{abstract}
  Lattice field theory is a very powerful tool to study Feynman's path integral non-perturbatively. However, it usually requires Euclidean background metrics to be well-defined. On the other hand, a recently developed regularization scheme based on Fourier integral operator $\zeta$-functions can treat Feynman's path integral non-pertubatively in Lorentzian background metrics. In this article, we formally $\zeta$-regularize lattice theories with Lorentzian backgrounds and identify conditions for the Fourier integral operator $\zeta$-function regularization to be applicable. Furthermore, we show that the classical limit of the $\zeta$-regularized theory is independent of the regularization. Finally, we consider the harmonic oscillator as an explicit example. We discuss multiple options for the regularization and analytically show that they all reproduce the correct ground state energy on the lattice and in the continuum limit. Additionally, we solve the harmonic oscillator on the lattice in Minkowski background numerically.  
\end{abstract}

\section{Introduction}

\subsection{Lattice Field Theory}

Feynman's path integral \cite{feynman,feynman-hibbs-styer} approach to 
quantum mechanics and quantum field theory has become 
an indispensable conceptual tool to evaluate physical observables, in 
particular in high energy physics. One very fascinating aspect 
of the path integral formulation is that it allows in principle
to compute also non-perturbative phenomena and thus reaches out 
to problems that can not be addressed in perturbation theory. 

The major obstacle to employ the path integral for studying 
non-perturbative effects is, however, that it is in general 
ill defined, if we take it at face value, i.e., using a Lorentzian  
background metric. An elegant way out is to analytically 
continue to imaginary -- Euclidean -- time and discretize the system 
either on a time lattice in quantum mechanics or on a space-time 
grid in quantum field theory. The so introduced lattice spacing 
acts then as an ultraviolet regulator, whereas 
a finite time or box length serves
as an infrared regulator resulting thus in completely regularized 
version of the path integral, see \cite{Gattringer:2010zz,Rothe:1992nt} 
for introductions to the 
so obtained {\em lattice field theories}. 

Another important advantage of the Euclidean lattice regulated path integral 
is its resemblance to systems in statistical physics where it 
corresponds to the partition function. This allows in particular  
to resort 
to numerical calculations using Markov Chain Monte Carlo (MCMC) methods to evaluate 
the path integral.
This approach has been employed very successfully in the last years 
to obtain many non-perturbative results mainly for quantum chromodynamics, the theory of the 
strong interaction between quarks and gluons. 
In particular, it became possible to perform an ab initio calculation of the 
low-lying baryon spectrum using only QCD as the underlying theory \cite{Durr:2008zz}. 
Furthermore, lattice QCD calculations led to a detailed insight into the 
structure of hadrons \cite{Constantinou:2015agp,Cichy:2018mum} and 
they have provided non-perturbative contributions to electroweak 
processes \cite{Meyer:2018til} and flavor physics \cite{Juettner:2016atf}, 
see e.g.~\cite{Aoki:2019cca,kronfeld2012twenty} for overviews.
They were also very successful to determine thermodynamic properties~\cite{Ding:2015ona}. 
The remarkable progress in lattice QCD simulations can be seen by the fact 
that nowadays 
calculations are performed on large lattices -- presently of the order of 
around $100^3\times 200$ lattice points -- and directly in physical conditions. 

Despite this impressive success, such Euclidean time lattice simulations can not
be employed in important and so far open questions such as the matter 
anti-matter asymmetry, topological theories, the discrepancy of the observed 
and theoretically predicted amount of charge-parity (CP) violation and 
non-equilibrium physics. 
The reason is that in these cases the integrand in the path integral 
becomes complex such that MCMC can not be used.

It would therefore be extremely valuable to have an alternative regularization
of the path integral which allows, on one hand, to address non-perturbative physics and,  
on the other hand, to use a Minkowski background metric or even more general Lorentzian 
metrics. It is precisely at this point where the proposed $\zeta$-regularization 
of the path integral \cite{hartung-jmp,hartung-iwota,hartung-jansen,jansen-hartung} is filling this so far missing gap.  
By ``gauging'' the path integral (cf. \autoref{sec:intro-zeta}),  
it allows for a fully non-perturbative regularization  
for very general metrics. 
In this way, physical observables can be computed
even in cases where 
the Euclidean lattice techniques such as MCMC simulations fail and even 
on quantum computers~\cite{hartung-jmp,hartung-iwota,hartung-jansen,jansen-hartung}.

In this article, we aim to extend the mathematical formulation of $\zeta$-regularization within the context of lattice theories and to illustrate its application at 
the example of the quantum mechanical harmonic oscillator. In particular, we will test different ways 
of gauging the path integral  and demonstrate that 
all the here employed choices lead to correct physical results. In addition, 
we will show that the $\zeta$-regularization gives the  
correct classical limit. 
Our work will therefore not only provide a rather simple example of the 
$\zeta$-regularization on the lattice with Lorentzian background metric, but it also discusses different classes of 
gauges providing thus suitable choices when applied to a particular physical model 
of interest.

\subsection{$\zeta$-regularization}\label{sec:intro-zeta}

Operator $\zeta$-functions are a means of constructing traces on certain operator algebras. Given an algebra $\Ap$ of operators, a trace $\tau$ that is defined on a sub-algebra $\Ap_0\sse\Ap$, and an operator $A\in\Ap\setminus\Ap_0$, the aim of operator $\zeta$-functions is to define $\tau(A)$. To this end, we construct a holomorphic family $\phi:\ \cn\to\Ap$ such that $\phi$ maps an open connected subset of $\cn$ into $\Ap_0$ and $\phi(0)=A$. Given this setup and sufficient regularity of the operator algebra and construction of $\phi$, we can study the operator $\zeta$-function $\zeta(\phi)$ which is defined via analytic continuation of $z\mapsto\tau(\phi(z))$ on the domain that is mapped to $\Ap_0$ under $\phi$. If this analytic extension exists in a neighborhood of $0$ and the value $\zeta(\phi)(0)$ is independent of the choices made in the construction of $\phi$, then we can use $\zeta(\phi)(0)$ as our definition of $\tau(\phi(0))$, i.e., $\tau(A)$.

For example, if we consider the absolute value of the differential operator $A:=\abs\d$ on $\rn/2\pi\zn$, then $\Ap$ is the algebra of classical pseudo-differential operators on $\rn/2\pi\zn$ and $\tau$ the trace on trace-class operators, i.e., the sum of eigenvalues counting multiplicities. Furthermore, we may construct $\phi(z):=A^{z+1}$. Formally evaluating the trace on $\phi(z)$ thus yields \mbox{$\tau(\phi(z))=\sum_{n\in\zn}\abs{n}^{z+1}=2\sum_{n\in\nn}n^{z+1}$.} This makes sense for $\Re(z)<-2$ and implies that the corresponding operator $\zeta$-function is $\zeta(\phi)(z)=2\zeta_R(-z-1)$ where $\zeta_R$ denotes the Riemann $\zeta$-function. In particular, the $\zeta$-regularized trace of $\abs\d$ is given by $\zeta(\phi)(0)=2\zeta_R(-1)=-\frac{1}{6}$.

Operator $\zeta$-functions were first studied for pseudo-differential operators by Ray and Singer~\cite{ray,ray-singer} whose work has been extended by many authors including Seeley~\cite{seeley}, Guillemin~\cite{guillemin-lagrangian,guillemin-residue-traces,guillemin-wave}, Kontsevich and Vishik~\cite{kontsevich-vishik,kontsevich-vishik-geometry}, Lesch~\cite{lesch}, Paycha and Scott~\cite{paycha,paycha-scott}, and Wodzicki~\cite{wodzicki}. This notion of $\zeta$-regularization has been introduced to physics by Hawking~\cite{hawking} in his work on path integrals with a curved space-time background, successfully applied in many physical settings (e.g., the Casimir effect, defining one-loop functional determinants, the stress-energy tensor, conformal field theory, and string theory~\cite{beneventano-santangelo,blau-visser-wipf,bordag-elizalde-kirsten,bytsenko-et-al,culumovic-et-al,dowker-critchley,elizalde2001,elizalde,elizalde-et-al,elizalde-vanzo-zerbini,fermi-pizzocchero,hawking,iso-murayama,marcolli-connes,mckeon-sherry,moretti97,moretti99,moretti00,moretti11,robles,shiekh,tong-strings}), and is related to Hadamard parametrix renormalization~\cite{hack-moretti}. As such the approach plays a foundational role for an effective Lagrangian to be defined~\cite{blau-visser-wipf}, for relatively easy computations of heat kernel coefficients~\cite{bordag-elizalde-kirsten}, and for non-trivial extensions of the Chowla-Selberg formula~\cite{elizalde2001}. Furthermore, the residues associated with operator $\zeta$-functions give rise to the multiplicative anomaly in perturbation theory~\cite{elizalde-vanzo-zerbini} and contribute to the energy momentum tensor of a black hole~\cite{hawking}.

Although $\zeta$-regularization based on pseudo-differential operators has been very successful, it cannot be applied in general. Radzikowski~\cite{radzikowski92,radzikowski96} showed that, in general, Fourier integral operator $\zeta$-functions~\cite{guillemin-lagrangian,guillemin-residue-traces,guillemin-wave,hartung-phd,hartung-scott} are necessary. The application of Fourier integral operator $\zeta$-functions to quantum field theory has recently been developed in the context of the Hamiltonian formulation of quantum field theories~\cite{hartung-phd,hartung-jmp,hartung-iwota,hartung-jansen,hartung-jansen-gauge-fields,hartung-scott,jansen-hartung}. The starting point remains fundamentally the same as with the pseudo-differential application of $\zeta$-regularization. Given an algebra of Fourier integral operators $\Ap$ corresponding to a quantum field theory, the time evolution operator $U$, and an observable $\Omega$, we would like to compute the vacuum expectation value
\begin{align*}\tag{$*$}\label{eq:feynman_definition_vacuum_expectation_value}
  \langle\Omega\rangle=\lim_{T\to\infty+i0}\frac{\tr(U(0,T)\Omega)}{\tr(U(0,T))}.
\end{align*}
Both $U(0,T)$ and $U(0,T)\Omega$ are Fourier integral operators in general, $U(0,T)$ is unitary, and $U(0,T)\Omega$ is usually unbounded. Hence, neither numerator nor denominator are well-defined. This justifies the Fourier integral operator $\zeta$-function approach to defining a $\zeta$-regularized vacuum expectation value which replaces the traces in \eref{eq:feynman_definition_vacuum_expectation_value} by operator $\zeta$-functions. To construct these $\zeta$-regularized vacuum expectation values~\cite{hartung-jmp,hartung-iwota,hartung-jansen,hartung-jansen-gauge-fields,jansen-hartung}, we construct a suitable holomorphic family of operators $\Gf$ called ``gauge''. This gauge has to satisfy a number of properties. Most importantly, for $\Re(z)\ll0$ both $U(0,T)\Gf(z)\Omega$ and $U(0,T)\Gf(z)$ have to be trace-class operators. This will imply that both $\zeta(U(0,T)\Gf\Omega)$ and $\zeta(U(0,T)\Gf)$ are well-defined meromorphic functions and therefore 
\begin{align*}
  \langle\Omega\rangle_\zeta(z)=\lim_{T\to\infty+i0}\frac{\zeta(U(0,T)\Gf\Omega)}{\zeta(U(0,T)\Gf)}(z)
\end{align*}
is a well-defined meromorphic function (provided that $\zeta(U(0,T)\Gf)$ is not the constant zero-function). In particular, if $\langle\Omega\rangle_\zeta$ is holomorphic in a neighborhood of $0$, then we can define the $\zeta$-regularized vacuum expectation value as $\langle\Omega\rangle_\zeta(0)$. This construction has been shown to be generally possible~\cite{hartung-jmp,hartung-iwota,hartung-jansen-gauge-fields}  and to be physically meaningful and accessible using quantum computing~\cite{hartung-jansen,jansen-hartung}.

Furthermore, the construction is formally applicable to lattice formulations in Lorentzian space-time as well. Applying a lattice discretization to $\zeta(U(0,T)\Gf\Omega)(z)$ for example yields integrals of the form $\int \omega(\Phi) \gf(z)(\Phi) e^{iS(\Phi)} \Dp(\Phi)$, i.e., precisely the type of integral we expect from a lattice discretization but with an additional term $\gf(z)(\Phi)$. This additional term $\gf(z)(\Phi)$ moreover satisfies the asymptotic bound $\abs{\gf(z)(\Phi)}\le C (1+\norm\Phi^2)^{\delta \Re(z)}$ for some constants $C,\delta>0$ and $\norm\Phi$ sufficiently large. Since $\omega$ is polynomially bounded for large $\norm\Phi$, this implies that $\Phi\mapsto\omega(\Phi) \gf(z)(\Phi) e^{iS(\Phi)}$ is integrable for $\Re(z)\ll0$ and the $\zeta$-regularized lattice vacuum expectation can be obtained analytically continuing to $z=0$.

If this general setup can be obtained from lattice discretizing $\zeta$-regularized vacuum expectation values starting from the Hamiltonian formulation, we may also take it as a starting point to an a priori $\zeta$-regularization of a lattice field theory without needing to construct the Hamiltonian formulation first. However, this raises a number of applicability questions. Most importantly:
\begin{enumerate}
\item[Q1:] What kind of gauge families $\gf$ can be used to $\zeta$-regularize lattice field theories in a Lorentzian background?
\item[Q2:] Under what conditions is it possible to construct a $\zeta$-regularized Hamiltonian theory starting from a ``na\"ively'' $\zeta$-regularized lattice field theory? 
\end{enumerate}
The first question therefore addresses the freedom we have in choosing gauge families most suitable for a given lattice simulation whereas answering the second question allows us to prove equivalence with the Fourier integral operator $\zeta$-function approach to $\zeta$-regularized vacuum expectation values. This is important because it means that the known applicability and physicality results~\cite{hartung-jmp,hartung-iwota,hartung-jansen,hartung-jansen-gauge-fields,jansen-hartung} extend to the thus constructed $\zeta$-regularized lattice field theory. In particular, it provides an avenue of using quantum computation to simulate lattice field theories in Lorentzian backgrounds.

\subsection{Aims of this article} Since $\zeta$-regularized lattice field theories in Lorentzian backgrounds have not been studied from this point of view, we aim to gain insights into questions Q1 and Q2. In \autoref{sec:zeta-reg-lattice} we will begin with a general discussion of lattice field theory from a distribution theory point of view. This is important for Q2 because the Hamiltonian theory uses microlocal analysis and therefore is fundamentally a theory of distributions. Formally introducing the gauge family $\gf$ in the second half of \autoref{sec:zeta-reg-lattice} we will show that the $\zeta$-regularized lattice field theory and the original lattice field theory coincide in the distributional limit $z\to0$. This is a necessary first step in the construction of gauged transfer matrices and gauged time evolution operators in \autoref{sec:gauging-transfer-matrix}. Thus, \autoref{sec:zeta-reg-lattice} addresses Q1 and \autoref{sec:gauging-transfer-matrix} addresses Q2. In \autoref{sec:classical-limit} we will show that $\zeta$-regularized lattice field theories have correct classical limits. In particular, we will explicitly compute the classical limit for the harmonic oscillator.  We will then continue discussing the harmonic oscillator analytically in \autoref{sec:harmonic_oscillator_analytic} and numerically in \autoref{sec:harmonic_oscillator_numeric}. For the analytic discussion, we will solve the harmonic oscillator five times, once with Euclidean background and then again with four different gauge families $\gf$. Each of these gauges will be representative of a generic property for the gauge that may be useful in different circumstances. The numerical discussion of the harmonic oscillator will focus on classical computing rather than quantum computing and highlight extrapolation and convergence behavior for one of the choices of gauge. 

\section{$\zeta$-regularizing lattice vacuum expectation values}\label{sec:zeta-reg-lattice}
In this section, we will investigate the connection between lattice vacuum expectation values with Euclidean and Lorentzian backgrounds. In particular, we will highlight how additional complications arise in Lorentzian backgrounds and how a distributional understanding of the lattice integrals and the $\zeta$-regularization are used to formally understand the lattice integrals in Lorentzian backgrounds. To this end, approximation on compact subsets will be instrumental. In \autoref{thm:compactification_limit} and \autoref{thm:analytic continuation}, we will provide conditions that ensure the infinite volume limit exists and coincides with the $\zeta$-regularized vacuum expectation value considered in~\cite{hartung-iwota,hartung-jmp,hartung-jansen,hartung-jansen-gauge-fields,jansen-hartung}. 

In a Euclidean lattice theory with $T$ time slices, we can express the vacuum expectation value~$\langle \Omega\rangle_E$ of an observable $\Omega$ as
\begin{align*}
  \langle \Omega\rangle_E=\frac{\int_{X^T}\Omega(\Phi)e^{-S(\Phi)}d\mu^T(\Phi)}{\int_{X^T}e^{-S(\Phi)}d\mu^T(\Phi)}
\end{align*}
where $X$ is a locally compact Hausdorff space, $\mu$ a $\sigma$-finite Radon measure on~$X$, and $\mu^T$ denotes the $T$-fold product measure $\mu\times\mu\times\cdots\times\mu$ on $X^T$.

\begin{example*}
  \begin{enumerate}
  \item[(i)] For the topological rotor (i.e., a massive particle on a circle of radius $R$), $X$ is the circle of radius $R$.
  \item[(ii)] For quantum mechanics in $\rn^N$, $X=\rn^N$.
  \item[(iii)] For a quantum field with values in $V$ on $L$ spacial lattice points, $X=V^L$. In particular, for a gauge field on $L$ spacial lattice points, $X$ often is $G^L$ for some (locally) compact Hausdorff group and $\mu$ the Haar measure on $X$.
  \end{enumerate}
\end{example*}

Vacuum expectation values of a lattice theory with Lorentzian background can be expressed as
\begin{align*}
  \langle \Omega\rangle_L=\frac{\int_{X^T}\Omega(\Phi)e^{iS(\Phi)}d\mu^T(\Phi)}{\int_{X^T}e^{iS(\Phi)}d\mu^T(\Phi)}
\end{align*}
where the integrals $\int_{X^T}F(\Phi)e^{iS(\Phi)}d\mu^T(\Phi)$ with $F\in\{1,\Omega\}$ are defined distributionally if $\mu(X)=\infty$. In terms of the transfer matrix, we can write 
\begin{align*}
  Z=\int_{X^T}e^{iS(\Phi)}d\mu^T(\Phi)=\int_{X^T}\prod_{j=1}^T e^{iS_\Delta(\Phi_{j+1},\Phi_j)}d\mu^T(\Phi)
\end{align*}
where the kernel $\kappa_\Delta$ of the transfer matrix\footnote{It should be noted that in complete generality, the transfer matrix may be ``time-dependent.'' If that is the case, then we have different $S_{\Delta,j}(\Phi_{j+1},\Phi_j)$, $\kappa_{\Delta,j}$, and $U_{\Delta,j}$. For the purpose of legibility, we will assume $\fa j,k:\ S_{\Delta,j}=S_{\Delta,k}=S_{\Delta}$ throughout this article. This assumption is not necessary for any of the results reported here to be true except that various powers, such as $U_\Delta^T$, would need to be replaced with the corresponding products, such as $\prod_{j=1}^TU_{\Delta,j}$. } $U_\Delta$ is given by $\kappa_\Delta(\Phi',\Phi)=e^{iS_\Delta(\Phi',\Phi)}$. In this distributional setting, $U_\Delta$ is therefore to be understood as the operator satisfying,
\begin{align*}\tag{$**$}\label{eq:definition U_Delta}
  \fa \phi,\psi\in C_c^\infty(X):\ \langle\phi,U_\Delta\psi\rangle=\int_X\int_X\phi(\Phi')^*e^{iS_\Delta(\Phi',\Phi)}\psi(\Phi)d\mu(\Phi')d\mu(\Phi).
\end{align*}

At this point we can make the connection to one of the necessary assumptions in the construction of $\zeta$-regularized vacuum expectation values in the Hamiltonian setting~\cite{hartung-iwota,hartung-jmp,hartung-jansen,hartung-jansen-gauge-fields,jansen-hartung}. For the approach to be applicable, we need to ensure that Cauchy surfaces (i.e., time slices) of the spacetime are compact. \eref{eq:definition U_Delta}, however,  implies that this assumption in the $\zeta$-formalism does not restrict generality. Let $K\sse X$ be compact. Then we define $U_\Delta|_K$ to be the restriction of $U_\Delta$ to $K$ which means that the kernel of $U_\Delta|_K$ is given by $1_{K\times K}\kappa_\Delta$ as a distribution on $K^2$ and $U_\Delta|_K$ is uniquely determined by
\begin{align*}
  \fa \phi,\psi\in C^\infty(K):\ \langle\phi,U_\Delta|_K\psi\rangle=\int_K\int_K\phi(\Phi')^*e^{iS_\Delta(\Phi',\Phi)}\psi(\Phi)d\mu|_K(\Phi')d\mu|_K(\Phi).
\end{align*}
The $\zeta$-formalism by construction only works for $U_\Delta|_K$. However, since for every $\phi,\psi\in C_c^\infty(X)$, there exists a compact $K_{\phi,\psi}\sse X$ such that $\spt\phi\cup\spt\psi\sse K_{\phi,\psi}$, where $\spt\phi$ denotes the support of $\phi$, we obtain that $\langle\phi,U_\Delta\psi\rangle=\langle\phi,U_\Delta|_{K_{\phi,\psi}}\psi\rangle$. Hence, $U_\Delta|_K$ converges to $U_\Delta$ for $K\nearrow X$ in the sense that the kernel $1_{K\times K}\kappa_\Delta$ of $U_\Delta|_K$ converges to the kernel $\kappa_\Delta$ of $U_\Delta$ in the space of distributions.

In the Euclidean setting, we generally don't have to discuss this aspect because the integrand $Fe^{-S}$ is in $L_1(X^2)$, and $\mu$ is a $\sigma$-finite Radon measure. This implies 
\begin{align*}
  &\int_{X^T}Fe^{-S}d\mu^T-\int_{K^T}Fe^{-S}d\mu^T=&\int_{X^T}1_{X^T\setminus K^T}Fe^{-S}d\mu^T\to&0\quad (K\nearrow X)
\end{align*}
and since $\int_{X^T}Fe^{-S}d\mu^T$ is a well-defined integral, the limit $K\nearrow X$ usually needs not to be considered.

In the $\zeta$-regularized setting, we enforce a similar situation by enforcing asymptotic bounds of the form $\abs{\Omega(\Phi)}\le c\norm\Phi^d$ for $\norm\Phi\to\infty$ and some constants $c,d\in\rn$ on the observable $\Omega$, and by introducing a suitably chosen holomorphic family of functions $\gf$ with $\gf(0)=1$ and $\abs{\gf(z)(\Phi)}\le \alpha\norm{\Phi}^{\beta\Re(z)}$ for constants $\alpha,\beta>0$ and $\norm{\Phi}\gg1$. This family $\gf$ is called a ``gauge'' in the mathematical literature and ensures that the gauged integrand $Fe^{iS}\gf(z)$ is in $L_1(X^T)$ for $\Re(z)\ll0$. We therefore obtain
\begin{align*}
  &\int_{X^T}Fe^{iS}\gf(z)d\mu^T-\int_{K^T}Fe^{iS}\gf(z)d\mu^T=&\int_{X^T}1_{X^T\setminus K^T}Fe^{iS}\gf(z)d\mu^T\to&0
\end{align*}
whenever $\Re(z)\ll0$. Thus, considering a sequence $(K_n)_{n\in\nn}$ of compacta in $X$ with $K_n\nearrow X\ (n\to\infty)$, we conclude
\begin{align*}
  \lim_{n\to\infty}\frac{\int_{K_n^T}\Omega e^{iS}\gf(z)d\mu^T}{\int_{K_n^T}e^{iS}\gf(z)d\mu^T}=\frac{\int_{X^T}\Omega e^{iS}\gf(z)d\mu^T}{\int_{X^T}e^{iS}\gf(z)d\mu^T}
\end{align*}
whenever $\Re(z)\ll0$. Since the $\zeta$-regularized vacuum expectation $\langle\Omega\rangle_\zeta$ is defined via analytic continuation of $z\mapsto\frac{\langle \Omega\gf(z)\rangle_L}{\langle\gf(z)\rangle_L}$, the above identity implies pointwise convergence on a half-space with $\Re(z)\ll0$.

\begin{theorem}\label{thm:compactification_limit}
  Let $(K_n)_{n\in\nn}$ be a sequence of compacta in $X$ with $K_n\nearrow X$ and $\l(\frac{\langle 1_{K_n^T}\Omega\gf\rangle_L}{\langle 1_{K_n^T}\gf\rangle_L}\r)_{n\in\nn}$ locally bounded in $C(D)$ where $D\sse\cn$ is open and connected containing a half-space with $\Re(z)\ll0$. Then, $\l(\frac{\langle 1_{K_n^T}\Omega\gf\rangle_L}{\langle 1_{K_n^T}\gf\rangle_L}\r)_{n\in\nn}$ converges compactly to the analytic continuation of $\frac{\langle\Omega\gf\rangle_L}{\langle\gf\rangle_L}$ on~$D$. In particular, if $0\in D$, we obtain
  \begin{align*}
    \lim_{n\to\infty}\langle 1_{K_n^T}\Omega\rangle_\zeta=\frac{\langle\Omega\gf\rangle_L}{\langle\gf\rangle_L}(0)=\langle \Omega\rangle_\zeta.
  \end{align*}
\end{theorem}
\autoref{thm:compactification_limit} follows directly from Vitali's Theorem since the set of pointwise convergence contains a half-space with $\Re(z)\ll0$ and we have asserted local boundedness.

\begin{remark*}
  The local boundedness assumption in \autoref{thm:compactification_limit} means that for every $z_0$ in the domain $D$, there exists a neighborhood $U\subseteq D$ of $z_0$ such that the sequence of functions $\frac{\langle 1_{K_n^T}\Omega\gf(z)\rangle_L}{\langle 1_{K_n^T}\gf(z)\rangle_L}$ with $z\in U$ is uniformly bounded as $K_n\nearrow X$, that is, 
  \begin{align*}
    \sup_{n\in\nn}\ \sup_{z\in U}\ \abs{\frac{\langle 1_{K_n^T}\Omega\gf(z)\rangle_L}{\langle 1_{K_n^T}\gf(z)\rangle_L}}<\infty.
  \end{align*}
  
  From a physical point of view, this means that any finite volume effects arising from the restriction to $K_n$ remain bounded under small variation of the gauge parameter~$z$. In this sense, the local boundedness assumption requires the physical theory to have locally bounded finite volume effects if the volume is sent to infinity and the observable is subject to holomorphic deformation. In particular, \autoref{thm:compactification_limit} then implies that locally bounded finite volume effects already imply vanishing finite volume effects in the infinite volume limit.
\end{remark*}

Since the $\zeta$-regularization requires compactification of $X$ and \autoref{thm:compactification_limit} ensures that we can reconstruct the non-compact case from compact restrictions, we will from now on only consider the case in which $X$ is compact. 
\begin{example*}
  For quantum mechanics in $\rn^N$, we can generally consider $X$ to be the flat torus $(\rn/L\zn)^N$ of side length $L\in\rn_{>0}$. In this case the limit $K\nearrow \rn^N$ corresponds to $L\nearrow\infty$.
\end{example*}

Returning to the transfer matrix $U_\Delta$ on compact $X$, we observe that the kernel 
\begin{align*}
  \kappa_{U_\Delta^T}(\Phi_{T+1},\Phi_1)=\int_{X^{T-1}}\prod_{j=1}^Te^{iS_\Delta(\Phi_{j+1},\Phi_j)}d\mu^{T-1}(\Phi_2,\ldots,\Phi_T)
\end{align*}
of $U_\Delta^T$ is integrable along the diagonal which implies that the nuclear trace 
\begin{align*}
  \tr(U_\Delta^T)=\langle\kappa_{U_\Delta^T},\delta_{\diag}\rangle=\int_X\kappa_{U_\Delta^T}(\Phi,\Phi)d\mu(\Phi)
\end{align*}
of $U_\Delta^T$ is well-defined. Introducing an observable $\Omega\in L_\infty(X)$ and gauge $\gf$, we thus obtain (using $F\in\{1,\Omega\}$) that
\begin{align*}
  \langle\kappa_{U_\Delta^T}F\gf(z),\delta_\diag\rangle=&\int_X\kappa_{U_\Delta^T}(\Phi,\Phi)F(\Phi)\gf(z)(\Phi)d\mu(\Phi)
\end{align*}
is a holomorphic function in $z$ which directly implies \autoref{thm:analytic continuation}.

\begin{theorem}\label{thm:analytic continuation}
  If $X$ is compact, then
  \begin{align*}
    \langle \Omega\rangle_L=\frac{\int_{X^T}\Omega(\Phi)e^{iS(\Phi)}d\mu^T(\Phi)}{\int_{X^T}e^{iS(\Phi)}d\mu^T(\Phi)}=\langle\Omega\rangle_\zeta.
  \end{align*}
  Furthermore, if $X$ is not compact, the assumptions of \autoref{thm:compactification_limit} are satisfied, and $(K_n)_{n\in\nn}$ is as in \autoref{thm:compactification_limit}, then
  \begin{align*}
    \lim_{n\to\infty}\frac{\int_{K_n^T}\Omega(\Phi)e^{iS(\Phi)}d\mu^T(\Phi)}{\int_{K_n^T}e^{iS(\Phi)}d\mu^T(\Phi)}=\langle \Omega\rangle_\zeta.
  \end{align*}
\end{theorem}

\section{Gauging the transfer matrix and time evolution operator}\label{sec:gauging-transfer-matrix}
In \autoref{sec:zeta-reg-lattice}, we have discussed how the $\zeta$-regularization can make sense of the formal integrals arising in a Lorentzian lattice theory through a process of gauging. However, the gauging process as discussed in \autoref{sec:zeta-reg-lattice} can be understood as an operation on the observable.  More precisely, as was shown in~\cite{hartung-jansen,jansen-hartung}, the entire construction of $\zeta$-regularized vacuum expectation values is equivalent to the construction of an analytic continuation of a quotient $\frac{\langle O_1(z)\rangle}{\langle O_2(z)\rangle}$ of expectation values $\langle O_j(z)\rangle$ for suitably constructed observables $O_j(z)$. In order to make the connection to the general $\zeta$-regularization theory of vacuum expectation values~\cite{hartung-jmp,hartung-iwota,hartung-jansen-gauge-fields}, on the other hand, we need to consider the gauge as part of the time evolution operator. Furthermore, for the $\zeta$-regularization point of view to be as versatile as methods used in Euclidean lattice theories, it would be advantageous to find gauges that can be attributed to the transfer matrix as opposed to the entire time evolution operator only. In this section, we will therefore discuss the relation between the point of view in \autoref{sec:zeta-reg-lattice} and the formalism of~\cite{hartung-iwota,hartung-jmp,hartung-jansen,hartung-jansen-gauge-fields,jansen-hartung}. As such, we will be able to conclude that all properties proven in~\cite{hartung-iwota,hartung-jmp,hartung-jansen,hartung-jansen-gauge-fields,jansen-hartung} also hold for the $\zeta$-regularization of lattice field theories. In particular, we will obtain that the $\zeta$-regularized lattice vacuum expectation values are independent of the choice of gauge $\gf$.

The transfer matrix $U_\Delta$ is given by its kernel $\kappa_\Delta(\Phi',\Phi)=e^{iS_\Delta(\Phi',\Phi)}$, i.e.,
\begin{align*}
  \fa \phi,\psi\in C_c^\infty(X):\ \langle\phi,U_\Delta\psi\rangle=\int_X\int_X\phi(\Phi')^*e^{iS_\Delta(\Phi',\Phi)}\psi(\Phi)d\mu(\Phi')d\mu(\Phi),
\end{align*}
and its adjoint $U_\Delta^*$ has kernel $\kappa_{U_\Delta^*}(\Phi',\Phi)=\kappa_\Delta(\Phi,\Phi')^*=e^{-iS_\Delta(\Phi,\Phi')}$. Since the transfer matrix connects neighboring values $\Phi_j$ and $\Phi_{j+1}$ in terms of the global action $e^{iS(\Phi)}=\prod_{j=1}^Te^{iS_\Delta(\Phi_{j+1},\Phi_j)}$, it can often be advantageous to mirror this structure in the choice of gauge $\gf$.  Introducing a gauge $\gf(z)(\Phi)=\prod_{j=1}^T\gf_\Delta(z)(\Phi_{j+1},\Phi_j)$, we obtain the gauged transfer matrix $\tilde U_\Delta$ defined as
\begin{align*}
  \l(\tilde U_\Delta(z)\psi\r)(\Phi')=\int_Xe^{iS_\Delta(\Phi',\Phi)}\gf_\Delta(z)(\Phi',\Phi)\psi(\Phi)d\mu(\Phi).
\end{align*}
This directly yields the gauge family $\Gf_\Delta$ in the operator point of view via 
\begin{align*}
  \l(\Gf_\Delta(z)\psi\r)(\Phi'')=&\l(U_\Delta^*\tilde U_\Delta(z)\psi\r)(\Phi'')\\
  =&\int_X\int_Xe^{-iS_\Delta(\Phi',\Phi'')}e^{iS_\Delta(\Phi',\Phi)}\gf_\Delta(z)(\Phi',\Phi)\psi(\Phi)d\mu(\Phi)d\mu(\Phi')
\end{align*}
and thus $\tilde U_\Delta=U_\Delta\Gf_\Delta$. In other words, $\Gf_\Delta(z)$ is an integral operator with kernel
\begin{align*}
  \kappa_{\Gf_\Delta(z)}(\Phi',\Phi)=\int_Xe^{-iS_\Delta(y,\Phi')}e^{iS_\Delta(y,\Phi)}\gf_\Delta(z)(y,\Phi)d\mu(y).
\end{align*}

\begin{theorem}\label{thm-gauged-transfer-matrix}
  Let $U:=U_\Delta^T$ be the time evolution operator defined by a transfer matrix $U_\Delta$ corresponding to an action $S_\Delta$, $\gf(z)(\Phi)=\prod_{j=1}^T\gf_\Delta(z)(\Phi_{j+1},\Phi_j)$, $\Gf_\Delta(z)$ the gauge family with the kernel
  \begin{align*}
    \kappa_{\Gf_\Delta(z)}(\Phi',\Phi)=\int_Xe^{-iS_\Delta(y,\Phi')}e^{iS_\Delta(y,\Phi)}\gf_\Delta(z)(y,\Phi)d\mu(y),
  \end{align*}
  and $\omega$ an observable with corresponding operator $\Omega$. Then, 
  \begin{align*}
    \langle\omega\rangle(z)=\frac{\int_{X^T}\omega(\Phi)e^{iS(\Phi)}\gf(z)(\Phi)d\mu^T(\Phi)}{\int_{X^T}e^{iS(\Phi)}\gf(z)(\Phi)d\mu^T(\Phi)}
  \end{align*}
  is the $\zeta$-regularized lattice vacuum expectation value which, in the operator formulation, corresponds to the $\zeta$-regularized vacuum expectation value 
  \begin{align*}
    \langle\Omega\rangle(z)=\frac{\zeta\l((U_\Delta\Gf_\Delta)^T\Omega\r)(z)}{\zeta\l((U_\Delta\Gf_\Delta)^T\r)(z)}.
  \end{align*}
  Since $(U_\Delta\Gf_\Delta)^T$ is a gauge of $U$ (more precisely, $(U_\Delta\Gf_\Delta)^T=U\Gf$ using the ``global'' gauge family $\Gf=U^*(U_\Delta\Gf_\Delta)^T$), we obtain gauge independence of $\langle\omega\rangle(0)$ in the sense of~\cite{hartung-iwota,hartung-jmp,hartung-jansen,hartung-jansen-gauge-fields,jansen-hartung}. In particular, if $U_\Delta$ and $\Gf_\Delta$ commute, then $\Gf=(\Gf_\Delta)^T$.
\end{theorem}

We note that the assumptions of \autoref{thm-gauged-transfer-matrix} permit gauge families $\gf$ of the form $\gf(z)(\Phi)=\l(\prod_{j=1}^T\abs{\Phi_j}\r)^{\frac{z}{T}}$ and $\gf(z)(\Phi)=\l(\prod_{j=1}^T\abs{\Phi_{j+1}-\Phi_j}\r)^{\frac{z}{T}}$ but not $\gf(z)(\Phi)=\norm\Phi^z$. The latter is a gauge for $U=U_\Delta^T$ rather than $U_\Delta$ itself. The kernel $\kappa_U$ of $U$ is given by
\begin{align*}
  \kappa_U(\Phi_{T+1},\Phi_1)=\int_{X^{T-1}}e^{iS(\Phi_{T+1},\ldots,\Phi_1)}d\mu^{T-1}(\Phi_T,\ldots,\Phi_2).
\end{align*}
Here, we do not identify $\Phi_{T+1}$ with $\Phi_1$ yet, since that step is eventually part of taking the trace. To reduce notation, we will introduce $\hat\Phi:=(\Phi_T,\ldots,\Phi_2)$ and $\hat\mu:=\mu^{T-1}$. Then, the action $S$ is a function of $(\Phi_{T+1},\hat\Phi,\Phi_1)$ due to the open boundary conditions and the gauge family $\gf$ is a family of functions of $(\hat\Phi,\Phi_1)$. Furthermore, $U^*$ has kernel
\begin{align*}
  \kappa_{U^*}(x,y)=\kappa_U(y,x)^*=\int_{X^{T-1}}e^{-iS(y,\hat\Phi,x)}d\hat\mu(\hat\Phi).
\end{align*}
Hence, the gauged time evolution $\tilde U$ with kernel
\begin{align*}
  \kappa_{\tilde U(z)}(\Phi_{T+1},\Phi_1)=\int_{X^{T-1}}e^{iS(\Phi_{T+1},\hat\Phi,\Phi_1)}\gf(z)(\hat\Phi,\Phi_1)d\hat\mu(\hat\Phi)
\end{align*}
can be expressed as $\tilde U=U\Gf$ with $\Gf=U^*\tilde U$, i.e., $\Gf(z)$ has kernel
\begin{align*}
  \kappa_{\Gf(z)}(x,y)=&\int_X\kappa_{U^*}(x,w)\kappa_{\tilde U(z)}(w,y)d\mu(w)\\
  =&\int_X\int_{X^{T-1}}e^{-iS(w,\hat\Phi,x)}d\hat\mu(\hat\Phi)\int_{X^{T-1}}e^{iS(w,\hat\Psi,y)}\gf(z)(\hat\Psi,y)d\hat\mu(\hat\Psi)d\mu(w).
\end{align*}
This directly implies \autoref{thm-gauged-time-evolution} which is the extension of \autoref{thm-gauged-transfer-matrix} to gauges $\gf$ that cannot be split into gauges of the transfer matrix.

\begin{theorem}\label{thm-gauged-time-evolution}
  Let $U$ be the time evolution operator corresponding to an action $S$, $\gf(z)$ a family of gauge functions, $\Gf(z)$ the gauge family with the kernel
  \begin{align*}
    \kappa_{\Gf(z)}(x,y)=\int_X\int_{X^{T-1}}\int_{X^{T-1}}e^{-iS(w,\hat\Phi,x)}e^{iS(w,\hat\Psi,y)}\gf(z)(\hat\Psi,y)d\hat\mu(\hat\Phi)d\hat\mu(\hat\Psi)d\mu(w),
  \end{align*}
  and $\omega$ an observable with corresponding operator $\Omega$. Then, 
  \begin{align*}
    \langle\omega\rangle(z)=\frac{\int_{X^T}\omega(\Phi)e^{iS(\Phi)}\gf(z)(\Phi)d\mu^T(\Phi)}{\int_{X^T}e^{iS(\Phi)}\gf(z)(\Phi)d\mu^T(\Phi)}
  \end{align*}
  is the $\zeta$-regularized lattice vacuum expectation value which, in the operator formulation, corresponds to the $\zeta$-regularized vacuum expectation value
  \begin{align*}
    \langle\Omega\rangle(z)=\frac{\zeta\l(U\Gf\Omega\r)(z)}{\zeta\l(U\Gf\r)(z)}.
  \end{align*}
  In particular, $\langle\omega\rangle(0)$ is independent of the gauge $\gf$ in the sense of~\cite{hartung-iwota,hartung-jmp,hartung-jansen,hartung-jansen-gauge-fields,jansen-hartung}.
\end{theorem}

This identifies sufficient conditions on the choice of gauge $\gf_\Delta$. Primarily, we need to ensure that $\Gf(z)$ is a Fourier integral operator of order $\delta\Re(z)$ for some \mbox{$\delta>0$} for the construction of $\zeta$-regularized vacuum expectation values to make sense in the operator picture~\cite{hartung-jmp,hartung-iwota,hartung-jansen,jansen-hartung}. A sufficient condition for this to be true is for $\gf(z)$ to be in the H\"ormander class $S^{\delta\Re(z)}$~\cite{hoermander-books}, that is, $\abs{\d_\Phi^\alpha\gf(z)(\Phi)}\le c_\alpha(1+\norm\Phi)^{\delta\Re(z)-\abs{\alpha}}$ has to hold for all $\Phi$, $z$, multiindices $\alpha$, and some constant $c_\alpha$. This condition can also be weakened by assuming that the inequality holds asymptotically for $\norm\Phi\to\infty$~\cite{hartung-phd,hartung-scott,hartung-jmp,hartung-iwota}. In doing so, we can allow for homogeneous singularities at $\Phi=0$, such as, $\gf(z)(\Phi)=\norm\Phi^z$ or even consider $\gf(z)$ to be polyhomogeneous $\gf(z)\sim\sum_{\iota\in I}a_\iota(z)(\Phi)$ where $a_\iota(z)$ is homogeneous of degree $d_\iota+z$ and the $\Re(d_\iota)$ are bounded from above~\cite{hartung-phd,hartung-scott,hartung-jmp,hartung-iwota}.

\section{The classical limit}\label{sec:classical-limit}
In the construction of $\zeta$-regularized vacuum expectation values, we consider a quotient of $\zeta$-functions $\frac{\zeta(U\Gf\Omega)}{\zeta(U\Gf)}$ which in the lattice setting takes the form $\frac{\int e^{iS}\gf\omega d\mu}{\int e^{iS}\gf d\mu}$. Although the general theory of operator $\zeta$-functions would allow for different choices of $\Gf$ (and thus different choices of $\gf$) in the numerator and denominator, the definition of $\zeta$-regularized vacuum expectation values~\cite{hartung-jmp} implicitly requires us to make the same choice. This was introduced to handle the problem that $\frac{\zeta(U\Gf_1\Omega)}{\zeta(U\Gf_2)}$ can depend on the choice of $\Gf_1$ and $\Gf_2$ if $\zeta(U\Gf_1\Omega)(0)=\zeta(U\Gf_2)(0)=0$. While enforcing $\Gf_1=\Gf_2=\Gf$ does not obviously solve the problem of $\frac{\zeta(U\Gf\Omega)}{\zeta(U\Gf)}(0)$ being dependent on $\Gf$ if $\zeta(U\Gf\Omega)(0)=\zeta(U\Gf)(0)=0$ -- in fact, independence of $\Gf$ is a consequence of the physicality proof for $\zeta$-regularized vacuum expectation values~\cite{hartung-jansen} -- there is a physical reason to hope for this to be the correct requirement. In a very loose sense, we can interpret the denominator $\zeta(U\Gf)(z)$ to be the partition function of a quantum field theory $\mathrm{QFT}(z)$ with ``time evolution'' $U\Gf(z)$. For $\Re(z)\ll0$ the vacuum expectation $\langle\Omega\rangle(z)$ of $\Omega$ is then well-defined in $\mathrm{QFT}(z)$ using Feynman's construction and satisfies $\langle\Omega\rangle(z)=\frac{\zeta(U\Gf\Omega)(z)}{\zeta(U\Gf)(z)}$. The $\zeta$-regularized vacuum expectation value $\langle\Omega\rangle(0)$ can therefore be loosely understood as the analytic continuation of well-defined vacuum expectation values in Feynman's sense with respect to the holomorphic family of quantum field theories $\mathrm{QFT}$.

Of course, this interpretation requires many philosophical concessions because in the lattice picture, $\mathrm{QFT}(z)$ corresponds to the ``action'' $S(z)=S(0)-i\ln(\gf(z))$ where $S(0)$ is the action of $\mathrm{QFT}(0)$, i.e., the action of the QFT we wish to study. However, if we accept this interpretation, then we can ask the question regarding the classical limit. In general, interpreting $S(z)$ as a classical action is likely to be difficult due to $-i\ln(\gf(z))$ being complex. The closest classical action we can extract from $S(z)$ is the original action $S(0)$ which leads to the hypothesis that the classical limit should always be the classical limit of $\mathrm{QFT}(0)$ independent of the value of $z$. \autoref{thm:classical-limit} is precisely this surprising result.

\begin{theorem}\label{thm:classical-limit}
  Let $S:\ \rn^N\to\rn$ be a non-degenerate action with a unique minimum, that is, there is exactly one point $C\in\rn^N$ with $S'(C)=0$ which furthermore satisfies $\det S''(C)\ne0$. Let $\omega:\ \rn^N\to\rn$ be an observable and $\gf(z)$ a gauge function with $\fa z\in\cn:\ \gf(z)(C)\ne0$. Then,
  \begin{align*}
    \fa z\in\cn:\ \lim_{\hbar\searrow0}\langle\omega\rangle(z)=\lim_{\hbar\searrow0}\frac{\int_{\rn^N}e^{\frac{i}{\hbar}S(\Phi)}\gf(z)(\Phi)\omega(\Phi)d\Phi}{\int_{\rn^N}e^{\frac{i}{\hbar}S(\Phi)}\gf(z)(\Phi)d\Phi}=\omega(C).
  \end{align*}
  In particular, $\lim_{\hbar\searrow0}\langle\omega\rangle(0)=\omega(C)$
\end{theorem}
\begin{proof}
  Since we are interested in the classical limit $\hbar\searrow0$, we will expand the integrals in
  \begin{align*}
    \langle\omega\rangle(z)=\frac{\int_{\rn^N}e^{\frac{i}{\hbar}S(\Phi)}\gf(z)(\Phi)\omega(\Phi)d\Phi}{\int_{\rn^N}e^{\frac{i}{\hbar}S(\Phi)}\gf(z)(\Phi)d\Phi}.
  \end{align*}
  using stationary phase approximation and obtain
  \begin{align*}
    \langle\omega\rangle(z)=&\frac{e^{\frac{i}{\hbar}S(C)}\abs{\det S''(C)}^{-\frac12}e^{\frac{i\pi}{4}\sgn S''(C)}\gf(z)(C)\omega(C)\l(2\pi\hbar\r)^{\frac{N}{2}}+o\l(\hbar^{\frac{N}{2}}\r)}{e^{\frac{i}{\hbar}S(C)}\abs{\det S''(C)}^{-\frac12}e^{\frac{i\pi}{4}\sgn S''(C)}\gf(z)(C)\l(2\pi\hbar\r)^{\frac{N}{2}}+o\l(\hbar^{\frac{N}{2}}\r)}\\
    =&\frac{\omega(C)+o(1)}{1+o(1)}.
  \end{align*}
  Hence, we observe
  \begin{align*}
    \fa z\in\cn:\ \lim_{\hbar\searrow0}\langle\omega\rangle(z)=\omega(C).
  \end{align*}
\end{proof}

\begin{example*}
  Consider the harmonic oscillator on the lattice with action
  \begin{align*}
    S:\ \rn^\tau\to\rn;\ x\mapsto a\sum_{j=0}^{\tau-1}\frac{m}{2}\frac{(x_{j+1}-x_j)^2}{a^2}+\mu^2x_j^2.
  \end{align*}
  Then
  \begin{align*}
    \d_kS(x)=&m\frac{x_{k}-x_{k-1}}{a}-m\frac{x_{k+1}-x_k}{a}+2a\mu^2x_k\\
    =&m\frac{2x_{k}-x_{k-1}-x_{k+1}}{a}+2a\mu^2x_k\\
    =&(Mx)_k
  \end{align*}
  with $M_{j,k}=\l(\frac{2m}{a}+2a\mu^2\r)\delta_{j,k}-\frac{m}{a}\delta_{j+1,k}-\frac{m}{a}\delta_{j-1,k}$. Thus, $M$ is a circulant matrix and setting $\Sigma_0=(\delta_{j+1,k})_{j,k}$ implies $M=\frac{2m}{a}+2a\mu^2-\frac{m}{a}\Sigma_0-\frac{m}{a}\Sigma_0^{\tau-1}=p(\Sigma_0)$. Since the eigenvalues of $\Sigma_0$ are $\lambda_k=e^{\frac{2\pi ik}{\tau}}$, the eigenvalues of $M$ are 
  \begin{align*}
    p(\lambda_k)=\frac{2m}{a}+2a\mu^2-\frac{m}{a}e^{\frac{2\pi ik}{\tau}}-\frac{m}{a}e^{-\frac{2\pi ik}{\tau}}=2a\mu^2+\frac{2m}{a}\l(1-\cos\l(\frac{2\pi k}{\tau}\r)\r).
  \end{align*}
  Hence, $S'(x)=0$ fails to have a unique solution ($C=0$) if and only if $\ex k:\ p(\lambda_k)=0$. The observation $p(\lambda_k)\ge2a\mu^2>0$ therefore implies that $C=0$ is the unique critical point, and we furthermore observe $S''(C)=M\in GL(N)$. Thus, the classical limit $\hbar\searrow0$ of $\langle\omega\rangle(z)$ is the observable $\omega$ evaluated on the classical vacuum $C=0$.
\end{example*}

\begin{remark*}
  The proof of \autoref{thm:classical-limit} can be extended to the case of multiple critical points $C$ with $S'(C)=0$ and $\det S''(C)\ne0$ such as in the case of a double well potential. Let $\Gamma$ be the set of critical points and $\omega$ and $\gf$ as in \autoref{thm:classical-limit}. Then, stationary phase approximation yields
  \begin{align*}
    \langle\omega\rangle(z)=&\frac{\sum_{C\in\Gamma}e^{\frac{i}{\hbar}S(C)}\abs{\det S''(C)}^{-\frac12}e^{\frac{i\pi}{4}\sgn S''(C)}\gf(z)(C)\omega(C)+o\l(1\r)}{\sum_{C\in\Gamma}e^{\frac{i}{\hbar}S(C)}\abs{\det S''(C)}^{-\frac12}e^{\frac{i\pi}{4}\sgn S''(C)}\gf(z)(C)+o\l(1\r)}
  \end{align*}
  and thus
  \begin{align*}
    \langle\omega\rangle(0)=&\frac{\sum_{C\in\Gamma}e^{\frac{i}{\hbar}S(C)}\abs{\det S''(C)}^{-\frac12}e^{\frac{i\pi}{4}\sgn S''(C)}\omega(C)+o\l(1\r)}{\sum_{C\in\Gamma}e^{\frac{i}{\hbar}S(C)}\abs{\det S''(C)}^{-\frac12}e^{\frac{i\pi}{4}\sgn S''(C)}+o\l(1\r)}.
  \end{align*}
  In the case of the double well potential, where we have exactly two critical points $C_1$ and $C_2$ with $S(C_1)=S(C_2)$ and $S''(C_1)=S''(C_2)$, we can further simplify to
  \begin{align*}
    \langle\omega\rangle(z)=&\frac{\gf(z)(C_1)\omega(C_1)+\gf(z)(C_2)\omega(C_2)+o\l(1\r)}{\gf(z)(C_1)+\gf(z)(C_2)+o\l(1\r)},\\
    \langle\omega\rangle(0)=&\frac{\omega(C_1)+\omega(C_2)+o\l(1\r)}{2+o\l(1\r)},
  \end{align*}
  and thus
  \begin{align*}
    \lim_{\hbar\to0}\langle\omega\rangle(0)=&\frac{\omega(C_1)+\omega(C_2)}{2}.
  \end{align*}
\end{remark*}

\section{The harmonic oscillator: four gauges and a Wick rotation}\label{sec:harmonic_oscillator_analytic}
In this section we will consider the harmonic oscillator and analytically compute the ground state energy with five different approaches. First (in \autoref{sec:harmonic-wick}), we will compute it using a Wick rotation; thus reproducing well-known results which will serve as a basis for comparison with the four different choices of gauge~$\gf$. These four choices of gauge arise from two binary choices. The first choice is whether we wish our gauge to be local in time or global, i.e., whether or not it is possible to split the gauge function $\gf(z)(\Phi)$ into a product $\prod_j\gf_\Delta(z)(\Phi_{j+1},\Phi_j)$ and thus gauging the transfer matrix $U_\Delta$. In applications, this choice will likely be a matter of convenience. The more important binary decision is related to the behavior of the gauge near zero. Since the regularization is introduced through the asymptotic behavior of the gauge for large $\norm\Phi$, there is a considerable degree of freedom in any fixed compact set. In particular, the functional analysis of a $\zeta$-lattice field theory is easier if the gauge is of the form $\gf(z)(\Phi)=\norm\Phi^z$ but this introduces a pole in zero. The corresponding integrals are still well-defined in terms of homogeneous distributions~\cite{hartung-iwota,hartung-jmp,hartung-jansen,hartung-jansen-gauge-fields,jansen-hartung} but not absolutely integrable. Thus, for numerical purposes, absolutely integrable gauges like $\gf(z)(\Phi)=(1+\norm\Phi)^z$ are preferable. In \autoref{sec:harmonic-global-distrib}, \autoref{sec:harmonic-global-absint}, \autoref{sec:harmonic-local-distrib}, and \autoref{sec:harmonic-local-absint} we will consider each of these cases (cf., \autoref{table:gauges}).

\begin{table}
  \caption{Different choices of gauge for the $\zeta$-regularized hamonic oscillator on the lattice with Lorentzian background.}\label{table:gauges}
  \begin{center}
    \begin{tabular}{|c|c|c|c|}
      \hline
      \multicolumn{2}{|c}{}&\multicolumn{2}{|c|}{regularization}\\
      \cline{3-4}
      \multicolumn{2}{|c|}{}&distributional&abs. integrable\\
      \hline
      \parbox[t]{2mm}{\multirow{6}{*}{\rotatebox[origin=c]{90}{gauging}}}&&&\\
      &global in time&\autoref{sec:harmonic-global-distrib}&\autoref{sec:harmonic-global-absint}\\
      &&&\\
      \cline{2-4}
      &&&\\
      &local in time&\autoref{sec:harmonic-local-distrib}&\autoref{sec:harmonic-local-absint}\\
      &&&\\
      \hline
    \end{tabular}
  \end{center}
\end{table}

In our treatment of the harmonic oscillator, we will absorb $\hbar$ into the action
\begin{align*}
  S:\ \rn^\tau\to\rn;\ \Phi\mapsto \frac{a}{\hbar}\sum_{j=0}^{\tau-1}\frac{m}{2}\frac{(\Phi_{j+1}-\Phi_j)^2}{a^2}+\frac{m\omega^2}{2}\Phi_j^2.
\end{align*}
Given an observable $\Omega$ and a gauge $\gf$, we hence want to solve
\begin{align*}
  \langle\Omega\rangle(z)=\frac{\int_{\rn^\tau}e^{iS(\Phi)}\gf(z)(\Phi)\Omega(\Phi)d\Phi}{\int_{\rn^\tau}e^{iS(\Phi)}\gf(z)(\Phi)d\Phi}
\end{align*}
in the $\zeta$-regularized cases and
\begin{align*}
  \langle\Omega\rangle_E=\frac{\int_{\rn^\tau}e^{-S(\Phi)}\Omega(\Phi)d\Phi}{\int_{\rn^\tau}e^{-S(\Phi)}d\Phi}
\end{align*}
in the Wick rotated case. For the calculations in this section, we will consider observables $\Omega$ of the form
\begin{align*}
  \Phi^\alpha=\frac{1}{\tau}\sum_{j=0}^{\tau-1}\Phi_j^\alpha.
\end{align*}
Many interesting observables are related to these $\Phi^\alpha$. For instance, the ground state energy is given by $iE_0=m\omega^2\langle\Phi^2\rangle(0)$ or $E_0=m\omega^2\langle\Phi^2\rangle_E$ where we have lost a factor of $i$ due to the Wick rotation.

We will compute the continuum limit $E_0^{\mathrm{continuum}}=\frac{\hbar\omega}{2}$ in the Wick rotated case explicitly. For all other cases, we will show that the ground state energy $E_0$ as computed with the $\zeta$-regularized lattice theory coincides with $E_0$ as computed with the Wick rotated lattice theory. To achieve this, we will define the matrix
\begin{align*}
  \Sigma:=\frac{1}{\hbar}\l(\l(\frac{m}{a}+\frac{am\omega^2}{2}\r)\delta_{j,k}-\frac{m}{2a}(\delta_{j,k+1}+\delta_{j+1,k})\r)_{j,k}
\end{align*}
because
\begin{align*}
  S(\Phi)=\langle\Phi,\Sigma\Phi\rangle.
\end{align*}
$\Sigma$ is a circulant matrix and setting $\Sigma_0=(\delta_{j+1,k})_{j,k}$ implies 
\begin{align*}
  \Sigma=\frac{m}{a\hbar}+\frac{am\omega^2}{2\hbar}-\frac{m}{2a\hbar}\Sigma_0-\frac{m}{2a\hbar}\Sigma_0^{\tau-1}=p(\Sigma_0).
\end{align*}
Since the eigenvalues of $\Sigma_0$ are $\mu_k=e^{\frac{2\pi ik}{\tau}}$, the eigenvalues of $\Sigma$ are 
\begin{align*}
  p(\mu_k)=\frac{m}{a\hbar}+\frac{am\omega^2}{2\hbar}-\frac{m}{2a\hbar}e^{\frac{2\pi ik}{\tau}}-\frac{m}{2a\hbar}e^{-\frac{2\pi ik}{\tau}}=\frac{am\omega^2}{2\hbar}+\frac{m}{a\hbar}\l(1-\cos\l(\frac{2\pi k}{\tau}\r)\r).
\end{align*}
Hence, $\Sigma$ is strictly positive and 
\begin{align*}
  S(\Phi)=\norm{\sqrt\Sigma\Phi}^2.
\end{align*}
In other words, the change of coordinates $\Psi:=\sqrt\Sigma\Phi$ implies
\begin{align*}
  \langle\Omega\rangle(z)=&\frac{\int_{\rn^\tau}e^{i\norm\Psi^2}\gf(z)\l(\sqrt\Sigma^{-1}\Psi\r)\Omega\l(\sqrt\Sigma^{-1}\Psi\r)d\Psi}{\int_{\rn^\tau}e^{i\norm\Psi^2}\gf(z)\l(\sqrt\Sigma^{-1}\Psi\r)d\Psi}\\
  =&\frac{\int_{\d B_{\rn^\tau}}\int_{\rn_{>0}}e^{ir^2}\gf(z)\l(r\sqrt\Sigma^{-1}\Psi\r)\Omega\l(r\sqrt\Sigma^{-1}\Psi\r)r^{\tau-1}drd\vol_{\d B_{\rn^\tau}}(\Psi)}{\int_{\d B_{\rn^\tau}}\int_{\rn_{>0}}e^{ir^2}\gf(z)\l(r\sqrt\Sigma^{-1}\Psi\r)r^{\tau-1}drd\vol_{\d B_{\rn^\tau}}(\Psi)}\\
  =&\frac{\int_{\d B_{\rn^\tau}}\int_{\rn_{>0}}e^{is}\gf(z)\l(\sqrt{s}\sqrt\Sigma^{-1}\Psi\r)\Omega\l(\sqrt{s}\sqrt\Sigma^{-1}\Psi\r)s^{\frac{\tau}{2}-1}dsd\vol_{\d B_{\rn^\tau}}(\Psi)}{\int_{\d B_{\rn^\tau}}\int_{\rn_{>0}}e^{is}\gf(z)\l(\sqrt{s}\sqrt\Sigma^{-1}\Psi\r)s^{\frac{\tau}{2}-1}dsd\vol_{\d B_{\rn^\tau}}(\Psi)}
\end{align*}
and
\begin{align*}
  \langle\Omega\rangle_E=\frac{\int_{\rn^\tau}e^{-\norm\Psi^2}\Omega\l(\sqrt\Sigma^{-1}\Psi\r)d\Psi}{\int_{\rn^\tau}e^{-\norm\Psi^2}d\Psi}
\end{align*}
where $\sqrt\Sigma^{-1}$ has the eigenvalues
\begin{align*}
  \lambda_k=\l(\frac{am\omega^2}{2\hbar}+\frac{m}{a\hbar}\l(1-\cos\l(\frac{2\pi k}{\tau}\r)\r)\r)^{-\frac{1}{2}}.
\end{align*}

\subsection{The Wick rotation}\label{sec:harmonic-wick}
We can compute $\langle\Omega\rangle_E$ by choosing an eigenbasis $(e_j)_j$ of $\sqrt\Sigma^{-1}$ such that $\Psi = \sum_j\Psi_je_j$ and $\sqrt\Sigma^{-1}e_j=\lambda_je_j$. Then, we obtain
\begin{align*}
  \langle\Omega\rangle_E=&\frac{\int_{\rn^\tau}e^{-\norm{\Psi}^2}\Omega\l(\sqrt\Sigma^{-1}\Psi\r)d\Psi}{\int_{\rn^\tau}e^{-\norm{\Psi}^2}d\Psi}\\
  =&\frac{\int_{\rn^\tau}\prod_{j=0}^{\tau-1}e^{-\Psi_j^2}\frac{1}{\tau}\sum_{k=0}^{\tau-1}(\lambda_k\Psi_k)^{\alpha}d\Psi}{\int_{\rn^\tau}\prod_{j=0}^{\tau-1}e^{-\Psi_j^2}d\Psi}\\
  =&\frac{1}{\tau}\sum_{k=0}^{\tau-1}\frac{\int_{\rn}e^{-\Psi_k^2}(\lambda_k\Psi_k)^{\alpha}d\Psi_k}{\int_{\rn}e^{-\Psi_k^2}d\Psi_k}\\
  =&\frac{1}{\tau}\sum_{k=0}^{\tau-1}\lambda_k^\alpha\frac{\l((-1)^\alpha+1\r)\Gamma\l(\frac{\alpha+1}{2}\r)}{2\Gamma\l(\frac{1}{2}\r)}.
\end{align*}
For $\alpha=2$, we conclude
\begin{align*}
  \langle\Phi^2\rangle_E=&\frac{1}{\tau}\sum_{k=0}^{\tau-1}\lambda_k^2\frac{\Gamma\l(\frac{1}{2}+1\r)}{\Gamma\l(\frac{1}{2}\r)}=\frac{1}{2\tau}\sum_{k=0}^{\tau-1}\lambda_k^2.
\end{align*}
To compute $\lim_{\tau\to\infty}\langle\Phi^2\rangle_E$, we use $\lim_{\tau\to\infty}\frac{1}{\tau}\sum_{k=0}^{\tau-1}f\l(\frac{2\pi k}{\tau}\r)=\frac{1}{2\pi}\int_0^{2\pi}f(x)dx$ to obtain\footnote{This $\tau\to\infty$ result coincides with the $N\to\infty$ limit of equation (C.29) in~\cite{Creutz-Freedman}. }
\begin{align*}
  \lim_{\tau\to\infty}\langle\Phi^2\rangle_E=&\frac{1}{4\pi}\int_0^{2\pi}\l(\frac{am\omega^2}{2\hbar}+\frac{m}{a\hbar}\l(1-\cos\l(x\r)\r)\r)^{-1}dx\\
  =&\l.\frac{\arctan\l(\sqrt{\frac{2m}{a\hbar}\frac{2\hbar}{am\omega^2}+1}\tan\l(\frac{x}{2}\r)\r)}{2\pi\sqrt{\frac{am\omega^2}{2\hbar}\l(\frac{am\omega^2}{2\hbar}+\frac{2m}{a\hbar}\r)}}\r|_{x=0}^{\pi}\\
  &+\l.\frac{\arctan\l(\sqrt{\frac{2m}{a\hbar}\frac{2\hbar}{am\omega^2}+1}\tan\l(\frac{x}{2}\r)\r)}{2\pi\sqrt{\frac{am\omega^2}{2\hbar}\l(\frac{am\omega^2}{2\hbar}+\frac{2m}{a\hbar}\r)}}\r|_{x=\pi}^{2\pi}\\
  =&\frac{\hbar}{2m\omega\sqrt{\frac{a^2\omega^2}{4}+1}}
\end{align*}
and thus
\begin{align*}
  E_0^{\mathrm{continuum}}=\lim_{a\to0}\lim_{\tau\to\infty}m\omega^2\langle\Phi^2\rangle_E=&\frac{\hbar\omega}{2}.
\end{align*}

For the $\zeta$-regularized computations to be consistent and independent of the choice of gauge, it suffices to show $\langle\Phi^2\rangle(0)=\frac{i}{2\tau}\sum_{k=0}^{\tau-1}\lambda_k^2$. If that is the case, we not only obtain the correct continuum limit for the ground state energy of the harmonic oscillator on the lattice, but the ground state energy coincides with the Wick rotated ground state energy result for all choices of lattice spacing $a$ and number of lattice sites $\tau$. 

\subsection{The global distributional gauge}\label{sec:harmonic-global-distrib}

The first type of gauge will be global in time and distributionally regularizing. These are generally of the form 
\begin{align*}
  \fa r\in\rn_{>0}\ \fa\Phi\in\d B_{\rn^\tau}:\ \gf(z)(r,\Phi)=r^z\gf_\d(z)(\Phi)=r^z\gf_\d(z)\l(\sqrt\Sigma^{-1}\Psi\r).
\end{align*}
Using such a gauge, we can evaluate the integral with respect to $s$ in 
\begin{align*}
  \langle\Omega\rangle(z)=&\frac{\int_{\d B_{\rn^\tau}}\int_{\rn_{>0}}e^{is}\gf(z)\l(\sqrt{s}\sqrt\Sigma^{-1}\Psi\r)\Omega\l(\sqrt{s}\sqrt\Sigma^{-1}\Psi\r)s^{\frac{\tau}{2}-1}dsd\vol_{\d B_{\rn^\tau}}(\Psi)}{\int_{\d B_{\rn^\tau}}\int_{\rn_{>0}}e^{is}\gf(z)\l(\sqrt{s}\sqrt\Sigma^{-1}\Psi\r)s^{\frac{\tau}{2}-1}dsd\vol_{\d B_{\rn^\tau}}(\Psi)}
\end{align*}
via analytic continuation of the Laplace transform $\Lp$ and obtain
\begin{align*}
  I(z):=&\int_{\d B_{\rn^\tau}}\int_{\rn_{>0}}e^{is}\gf(z)\l(\sqrt{s}\sqrt\Sigma^{-1}\Psi\r)\Omega\l(\sqrt{s}\sqrt\Sigma^{-1}\Psi\r)s^{\frac{\tau}{2}-1}dsd\vol_{\d B_{\rn^\tau}}(\Psi)\\
  =&\int_{\d B_{\rn^\tau}}\int_{\rn_{>0}}e^{is}\gf_\d(z)\l(\sqrt\Sigma^{-1}\Psi\r)\Omega\l(\sqrt\Sigma^{-1}\Psi\r)s^{\frac{\tau+z+\alpha}{2}-1}dsd\vol_{\d B_{\rn^\tau}}(\Psi)\\
  =&\int_{\d B_{\rn^\tau}}\gf_\d(z)\l(\sqrt\Sigma^{-1}\Psi\r)\Omega\l(\sqrt\Sigma^{-1}\Psi\r)\Lp\l(s\mapsto s^{\frac{\tau+z+\alpha}{2}-1}\r)(-i)d\vol_{\d B_{\rn^\tau}}(\Psi)\\
  =&i^{\frac{\tau+z+\alpha}{2}}\Gamma\l(\frac{\tau+z+\alpha}{2}\r)\int_{\d B_{\rn^\tau}}\gf_\d(z)\l(\sqrt\Sigma^{-1}\Psi\r)\Omega\l(\sqrt\Sigma^{-1}\Psi\r)d\vol_{\d B_{\rn^\tau}}(\Psi).
\end{align*}
Hence,
\begin{align*}
  \langle\Omega\rangle(z)=&\frac{i^{\frac{\tau+z+\alpha}{2}}\Gamma\l(\frac{\tau+z+\alpha}{2}\r)\int_{\d B_{\rn^\tau}}\gf_\d(z)\l(\sqrt\Sigma^{-1}\Psi\r)\Omega\l(\sqrt\Sigma^{-1}\Psi\r)d\vol_{\d B_{\rn^\tau}}(\Psi)}{i^{\frac{\tau+z}{2}}\Gamma\l(\frac{\tau+z}{2}\r)\int_{\d B_{\rn^\tau}}\gf_\d(z)\l(\sqrt\Sigma^{-1}\Psi\r)d\vol_{\d B_{\rn^\tau}}(\Psi)}\\
  =&\frac{i^{\frac{\alpha}{2}}\Gamma\l(\frac{\tau+z+\alpha}{2}\r)\int_{\d B_{\rn^\tau}}\gf_\d(z)\l(\sqrt\Sigma^{-1}\Psi\r)\Omega\l(\sqrt\Sigma^{-1}\Psi\r)d\vol_{\d B_{\rn^\tau}}(\Psi)}{\Gamma\l(\frac{\tau+z}{2}\r)\int_{\d B_{\rn^\tau}}\gf_\d(z)\l(\sqrt\Sigma^{-1}\Psi\r)d\vol_{\d B_{\rn^\tau}}(\Psi)}
\end{align*}
and
\begin{align*}
  \langle\Omega\rangle(0)=&\frac{i^{\frac{\alpha}{2}}\Gamma\l(\frac{\tau+\alpha}{2}\r)\int_{\d B_{\rn^\tau}}\Omega\l(\sqrt\Sigma^{-1}\Psi\r)d\vol_{\d B_{\rn^\tau}}(\Psi)}{\Gamma\l(\frac{\tau}{2}\r)\vol_{\d B_{\rn^\tau}}(\d B_{\rn^\tau})}.
\end{align*}
If we choose the eigenbasis of $\Sigma$ to parametrize $\rn^\tau$ again, then 
\begin{align*}
  \int_{\d B_{\rn^\tau}}\Omega\l(\sqrt\Sigma^{-1}\Psi\r)d\vol_{\d B_{\rn^\tau}}(\Psi)=&\frac1\tau\sum_{j=0}^{\tau-1}\lambda_j^\alpha\int_{\d B_{\rn^\tau}}\Psi_j^\alpha d\vol_{\d B_{\rn^\tau}}(\Psi)
\end{align*}
needs to be evaluated. In this case, we can perform the calculation analytically
\begin{align*}
  \int_{\d B_{\rn^\tau}}\Psi_j^\alpha d\vol_{\d B_{\rn^\tau}}(\Psi)=&\int_0^\pi\ldots\int_0^\pi\int_0^{2\pi}\cos(\phi_1)^\alpha\prod_{j=1}^{\tau-2}\sin(\phi_j)^{\tau-1-j}d\phi_{\tau-1}\ldots d\phi_1\\
  =&\vol_{\d B_{\rn^{\tau-1}}}\l(\d B_{\rn^{\tau-1}}\r)\int_0^\pi\cos(\phi)^\alpha\sin(\phi)^{\tau-2}d\phi\\
  =&\vol_{\d B_{\rn^{\tau-1}}}\l(\d B_{\rn^{\tau-1}}\r)\frac{\l((-1)^\alpha+1\r)\Gamma\l(\frac{\alpha+3}{2}\r)\Gamma\l(\frac{\tau-1}{2}\r)}{(\alpha+1)\Gamma\l(\frac{\tau+\alpha}{2}\r)}
\end{align*}
which yields
\begin{align*}
  \langle\Omega\rangle(0)=&\frac1\tau\sum_{j=0}^{\tau-1}\lambda_j^\alpha\frac{i^{\frac{\alpha}{2}}\Gamma\l(\frac{\tau+\alpha}{2}\r)\vol_{\d B_{\rn^{\tau-1}}}\l(\d B_{\rn^{\tau-1}}\r)\l((-1)^\alpha+1\r)\Gamma\l(\frac{\alpha+3}{2}\r)\Gamma\l(\frac{\tau-1}{2}\r)}{\Gamma\l(\frac{\tau}{2}\r)\vol_{\d B_{\rn^\tau}}(\d B_{\rn^\tau})(\alpha+1)\Gamma\l(\frac{\tau+\alpha}{2}\r)}\\
  =&\frac{i^{\frac{\alpha}{2}}}{(\alpha+1)\tau}\sum_{j=0}^{\tau-1}\lambda_j^\alpha\frac{\vol_{\d B_{\rn^{\tau-1}}}\l(\d B_{\rn^{\tau-1}}\r)\l((-1)^\alpha+1\r)\Gamma\l(\frac{\alpha+3}{2}\r)\Gamma\l(\frac{\tau-1}{2}\r)}{\Gamma\l(\frac{\tau}{2}\r)\vol_{\d B_{\rn^\tau}}(\d B_{\rn^\tau})}\\
  =&\frac{i^{\frac{\alpha}{2}}\l((-1)^\alpha+1\r)\Gamma\l(\frac{\alpha+3}{2}\r)}{(\alpha+1)\sqrt\pi\tau}\sum_{j=0}^{\tau-1}\lambda_j^\alpha\\
  =&\frac{i^{\frac{\alpha}{2}}\l((-1)^\alpha+1\r)\Gamma\l(\frac{\alpha+1}{2}\r)}{2\sqrt\pi\tau}\sum_{j=0}^{\tau-1}\lambda_j^\alpha.
\end{align*}
In particular, for $\alpha=2$ we obtain
\begin{align*}
  \langle\Phi^2\rangle(0)=&\frac{i\Gamma\l(\frac{3}{2}\r)}{\sqrt\pi\tau}\sum_{j=0}^{\tau-1}\lambda_j^2=\frac{i}{2\tau}\sum_{j=0}^{\tau-1}\lambda_j^2
\end{align*}
and thus the correct ground state energy.

\subsection{The global absolutely integrable gauge}\label{sec:harmonic-global-absint}

For the gauge to yield an absolutely integrable regularization, we need to remove the pole in $0$. For instance, a gauge of the form 
\begin{align*}
  \fa r\in\rn_{>0}\ \fa\Phi\in\d B_{\rn^\tau}:\ \gf(z)(r,\Phi)=\begin{cases}1&,\ r\le1\\\gf_\d(z)(\Phi)r^z&,\ r>1\end{cases}
\end{align*}
may be chosen. Then
\begin{align*}
  I(z):=&\int_{\d B_{\rn^\tau}}\int_{\rn_{>0}}e^{is}\gf(z)\l(\sqrt{s}\sqrt\Sigma^{-1}\Psi\r)\Omega\l(\sqrt{s}\sqrt\Sigma^{-1}\Psi\r)s^{\frac{\tau}{2}-1}dsd\vol_{\d B_{\rn^\tau}}(\Psi)\\
  =&\int_{\d B_{\rn^\tau}}\Omega\l(\sqrt\Sigma^{-1}\Psi\r)d\vol_{\d B_{\rn^\tau}}(\Psi)\int_0^1e^{is}s^{\frac{\tau+\alpha}{2}-1}ds\\
  &+\int_{\d B_{\rn^\tau}}\gf_\d(z)\l(\sqrt\Sigma^{-1}\Psi\r)\Omega\l(\sqrt\Sigma^{-1}\Psi\r)d\vol_{\d B_{\rn^\tau}}(\Psi)\int_{\rn_{>1}}e^{is}s^{\frac{\tau+z+\alpha}{2}-1}ds\\
  =&\int_{\d B_{\rn^\tau}}\Omega\l(\sqrt\Sigma^{-1}\Psi\r)d\vol_{\d B_{\rn^\tau}}(\Psi)i^{\frac{\tau+\alpha}{2}}\gamma\l(\frac{\tau+\alpha}{2},-i\r)\\
  &+\int_{\d B_{\rn^\tau}}\gf_\d(z)\l(\sqrt\Sigma^{-1}\Psi\r)\Omega\l(\sqrt\Sigma^{-1}\Psi\r)d\vol_{\d B_{\rn^\tau}}(\Psi)i^{\frac{\tau+z+\alpha}{2}}\Gamma\l(\frac{\tau+z+\alpha}{2},-i\r).
\end{align*}
In particular, for $\gf_\d(z)=1$ we obtain
\begin{align*}
  I(z)=&\int_{\d B_{\rn^\tau}}\Omega\l(\sqrt\Sigma^{-1}\Psi\r)d\vol_{\d B_{\rn^\tau}}(\Psi)i^{\frac{\tau+\alpha}{2}}\l(\gamma\l(\frac{\tau+\alpha}{2},-i\r)+\Gamma\l(\frac{\tau+z+\alpha}{2},-i\r)i^{\frac{z}{2}}\r)
\end{align*}
and thus
\begin{align*}
  \langle\Omega\rangle(z)=\frac{\int_{\d B_{\rn^\tau}}\Omega\l(\sqrt\Sigma^{-1}\Psi\r)d\vol_{\d B_{\rn^\tau}}(\Psi)}{\vol_{\d B_{\rn^\tau}}(\d B_{\rn^\tau})}\frac{i^{\frac{\alpha}{2}}\l(\gamma\l(\frac{\tau+\alpha}{2},-i\r)+\Gamma\l(\frac{\tau+z+\alpha}{2},-i\r)i^{\frac{z}{2}}\r)}{\l(\gamma\l(\frac{\tau}{2},-i\r)+\Gamma\l(\frac{\tau+z}{2},-i\r)i^{\frac{z}{2}}\r)}.
\end{align*}
At $z=0$, we may use
\begin{align*}
  \frac{i^{\frac{\alpha}{2}}\l(\gamma\l(\frac{\tau+\alpha}{2},-i\r)+\Gamma\l(\frac{\tau+\alpha}{2},-i\r)\r)}{\l(\gamma\l(\frac{\tau}{2},-i\r)+\Gamma\l(\frac{\tau}{2},-i\r)\r)}=\frac{i^{\frac{\alpha}{2}}\Gamma\l(\frac{\tau+\alpha}{2}\r)}{\Gamma\l(\frac{\tau}{2}\r)}
\end{align*}
and comparing the remaining integrals to the results of \autoref{sec:harmonic-global-distrib}, we observe that the two choices of gauge give the same result (as expected since we have gauge invariance). At the same time this appears a lot more complicated than the computation in \autoref{sec:harmonic-global-distrib}. However, there is a major advantage in that the initial $\rn^\tau$ integrals can be evaluated numerically for $\Re(z)\ll0$ without using the Laplace transform. Hence, we can compute $\frac{\int_{\d B_{\rn^\tau}}\Omega\l(\sqrt\Sigma^{-1}\Psi\r)d\vol_{\d B_{\rn^\tau}}(\Psi)}{\vol_{\d B_{\rn^\tau}}(\d B_{\rn^\tau})}$, which requires integration over high dimensional spheres, by computing
\begin{align*}
  \frac{\gamma\l(\frac{\tau}{2},-i\r)+\Gamma\l(\frac{\tau+z}{2},-i\r)i^{\frac{z}{2}}}{i^{\frac{\alpha}{2}}\l(\gamma\l(\frac{\tau+\alpha}{2},-i\r)+\Gamma\l(\frac{\tau+z+\alpha}{2},-i\r)i^{\frac{z}{2}}\r)}\frac{\int_{\rn^\tau}e^{i\norm{\Psi}^2}\gf(z)\l(\sqrt\Sigma^{-1}\Psi\r)\Omega\l(\sqrt\Sigma^{-1}\Psi\r)d\Psi}{\int_{\rn^\tau}e^{i\norm{\Psi}^2}\gf(z)\l(\sqrt\Sigma^{-1}\Psi\r)d\Psi}
\end{align*}
for values of $z$ with $\Re(z)\ll0$. This may be easier to implement as more numerical techniques have been developed for integrating over $\rn^\tau$ than $\d B_{\rn^\tau}$.

\subsection{The local distributional gauge}\label{sec:harmonic-local-distrib}
With the third approach we want to gauge the transfer matrix (as opposed to ``only'' gauging the entire time evolution operator). This is often useful because it retains locality in time -- a property frequently exploited in many lattice methods; e.g., when computing correlation functions. In terms of the following computation, this choice of gauge allows us to use Fubini as in the Wick rotated computation. Let us start with a generic $\gf(z)(\Phi)=\prod_{j=0}^{\tau-1}\gf_\Delta(\Phi_j)$. Then we observe
\begin{align*}
  \langle\Omega\rangle(z)=&\frac{\int_{\rn^\tau}e^{i\norm{\Psi}^2}\gf(z)\l(\sqrt\Sigma^{-1}\Psi\r)\Omega\l(\sqrt\Sigma^{-1}\Psi\r)d\Psi}{\int_{\rn^\tau}e^{i\norm{\Psi}^2}\gf(z)\l(\sqrt\Sigma^{-1}\Psi\r)d\Psi}\\
  =&\frac{\int_{\rn^\tau}\prod_{j=0}^{\tau-1}e^{i\Psi_j^2}\gf_\Delta(z)\l(\lambda_j\Psi_j\r)\frac{1}{\tau}\sum_{k=0}^{\tau-1}(\lambda_k\Psi_k)^{\alpha}d\Psi}{\int_{\rn^\tau}\prod_{j=0}^{\tau-1}e^{i\Psi_j^2}\gf_\Delta(z)\l(\lambda_j\Psi_j\r)d\Psi}\\
  =&\frac{1}{\tau}\sum_{k=0}^{\tau-1}\frac{\int_{\rn}e^{i\Psi_k^2}\gf_\Delta(z)\l(\lambda_k\Psi_k\r)(\lambda_k\Psi_k)^{\alpha}d\Psi_k}{\int_{\rn}e^{i\Psi_k^2}\gf_\Delta(z)\l(\lambda_k\Psi_k\r)d\Psi_k}
\end{align*}
where
\begin{align*}
  \lambda_j=\l(\frac{am\omega^2}{2}+\frac{m}{a}\l(1-\cos\l(\frac{2\pi j}{\tau}\r)\r)\r)^{-\frac12}.
\end{align*}
If we choose $\gf_\Delta(z)\l(\lambda_k\Psi_k\r)=\gf_\Delta(z)\l(\lambda_k\abs{\Psi_k}\r)$, we further obtain
\begin{align*}
  \langle\Omega\rangle(z)=&\frac{1}{\tau}\sum_{k=0}^{\tau-1}\lambda_k^\alpha\frac{\int_{\rn_{>0}}e^{i\psi}\gf_\Delta(z)\l(\lambda_k\sqrt{\Psi_k}\r)\Psi_k^{\frac{\alpha-1}{2}}d\Psi_k}{\int_{\rn_{>0}}e^{i\psi}\gf_\Delta(z)\l(\lambda_k\sqrt{\Psi_k}\r)\Psi_k^{-\frac{1}{2}}d\Psi_k}.
\end{align*}
Again, we have the choice between a distributionally regularizing gauge or an absolutely integrable gauge. In this subsection, we will proceed with a distributionally regularizing gauge $\gf_\Delta(z)(\psi)=\abs\psi^z$ which implies
\begin{align*}
  \langle\Omega\rangle(z)=&\frac{1}{\tau}\sum_{k=0}^{\tau-1}\lambda_k^{\alpha}\frac{\int_{\rn_{>0}}e^{i\psi}\psi^{\frac{\alpha+z-1}{2}}d\psi}{\int_{\rn_{>0}}e^{i\psi}\psi^{\frac{z-1}{2}}d\psi}
  =\frac{i^{\frac{\alpha}{2}}\Gamma\l(\frac{\alpha+z+1}{2}\r)}{\Gamma\l(\frac{z+1}{2}\r)\tau}\sum_{k=0}^{\tau-1}\lambda_k^{\alpha}
\end{align*}
and hence the required
\begin{align*}
  \langle\Phi^2\rangle(0)=&\lim_{z\to0}\frac{i\Gamma\l(\frac{z+1}{2}+1\r)}{\Gamma\l(\frac{z+1}{2}\r)\tau}\sum_{k=0}^{\tau-1}\lambda_k^2
  =\frac{i}{2\tau}\sum_{k=0}^{\tau-1}\lambda_k^2.
\end{align*}

\subsection{The local absolutely integrable gauge}\label{sec:harmonic-local-absint}
The fourth and final choice of gauge would be an absolutely integrable gauge for the transfer matrix. For example, we may consider
\begin{align*}
  \gf_\Delta(z)(\psi)=\begin{cases}1&,\ \abs\psi<1\\\abs\psi^z&,\ \abs\psi\ge1\end{cases}.
\end{align*}
In this case, we obtain
\begin{align*}
  \langle\Omega\rangle(z)=&\frac{1}{\tau}\sum_{k=0}^{\tau-1}\lambda_k^\alpha\frac{\int_{\rn_{>0}}e^{i\psi}\gf_\Delta(z)\l(\lambda_k\sqrt\psi\r)\psi^{\frac{\alpha-1}{2}}d\psi}{\int_{\rn_{>0}}e^{i\psi}\gf_\Delta(z)\l(\lambda_k\sqrt\psi\r)\psi^{-\frac{1}{2}}d\psi}\\
  =&\frac{1}{\tau}\sum_{k=0}^{\tau-1}\lambda_k^\alpha\frac{\int_0^{\lambda_k^{-1}}e^{i\psi}\psi^{\frac{\alpha-1}{2}}d\psi+\lambda_k^z\int_{\rn_{>{\lambda_k^{-1}}}}e^{i\psi}\psi^{\frac{\alpha+z-1}{2}}d\psi}{\int_0^{\lambda_k^{-1}}e^{i\psi}\psi^{-\frac{1}{2}}d\psi+\lambda_k^z\int_{\rn_{>{\lambda_k^{-1}}}}e^{i\psi}\psi^{\frac{z-1}{2}}d\psi}\\
  =&\frac{1}{\tau}\sum_{k=0}^{\tau-1}\lambda_k^\alpha\frac{i^{\frac{\alpha+1}{2}}\gamma\l(\frac{\alpha+1}{2},\lambda_k^{-1}\r)+\lambda_k^zi^{\frac{\alpha+z+1}{2}}\Gamma\l(\frac{\alpha+z+1}{2},\lambda_k^{-1}\r)}{i^{\frac{1}{2}}\gamma\l(\frac{1}{2},\lambda_k^{-1}\r)+\lambda_k^zi^{\frac{z+1}{2}}\Gamma\l(\frac{z+1}{2},\lambda_k^{-1}\r)}
\end{align*}
which in the limit $z\to0$ yields
\begin{align*}
  \langle\Omega\rangle(0)=&\frac{1}{\tau}\sum_{k=0}^{\tau-1}\lambda_k^\alpha\frac{i^{\frac{\alpha}{2}}\gamma\l(\frac{\alpha+1}{2},\lambda_k^{-1}\r)+i^{\frac{\alpha}{2}}\Gamma\l(\frac{\alpha+1}{2},\lambda_k^{-1}\r)}{\gamma\l(\frac{1}{2},\lambda_k^{-1}\r)+\Gamma\l(\frac{1}{2},\lambda_k^{-1}\r)}
  =\frac{i^{\frac{\alpha}{2}}}{\tau}\sum_{k=0}^{\tau-1}\lambda_k^\alpha\frac{\Gamma\l(\frac{\alpha+1}{2}\r)}{\Gamma\l(\frac{1}{2}\r)}
\end{align*}
and thus
\begin{align*}
  \langle\Phi^2\rangle(0)=&\frac{i}{2\tau}\sum_{k=0}^{\tau-1}\lambda_k^2.
\end{align*}

\section{The harmonic oscillator numerically}\label{sec:harmonic_oscillator_numeric}
While the harmonic oscillator can be solved analytically, as seen in \autoref{sec:harmonic_oscillator_analytic}, more complicated lattice theories will need to be solved numerically. To this end, we may choose between a quantum simulation of the lattice theory~\cite{hartung-jansen,jansen-hartung} or a classical simulation. In this section, we will therefore consider the numerical approach to the $\zeta$-regularized harmonic oscillator using classical computation as an example for $\zeta$-regularized lattice theories in detail.

If we want to treat the $\zeta$-regularized harmonic oscillator numerically, we have in principle two options. The first option is to choose a gauge $\gf$ and compute the dependence on the gauge parameter $z$ explicitly. In this case, we may choose a distributionally regularizing gauge $\gf(z)(r,\Phi)=r^z\gf_\d(z)(\Phi)$ and obtain an expression of the form
\begin{align*}
  \langle\Omega\rangle(z)=&\frac{i^{\frac{\alpha}{2}}\Gamma\l(\frac{\tau+z+\alpha}{2}\r)\int_{\d B_{\rn^\tau}}\gf_\d(z)\l(\sqrt\Sigma^{-1}\Psi\r)\Omega\l(\sqrt\Sigma^{-1}\Psi\r)d\vol_{\d B_{\rn^\tau}}(\Psi)}{\Gamma\l(\frac{\tau+z}{2}\r)\int_{\d B_{\rn^\tau}}\gf_\d(z)\l(\sqrt\Sigma^{-1}\Psi\r)d\vol_{\d B_{\rn^\tau}}(\Psi)}.
\end{align*}
This enables us to take the $z\to0$ limit explicitly and we are left with the numerical problem of integrating
\begin{align*}
  E_0 =&\frac{m\omega^2\tau}{2}\frac{\int_{\d B_{\rn^\tau}}\Omega\l(\sqrt\Sigma^{-1}\Psi\r)d\vol_{\d B_{\rn^\tau}}(\Psi)}{\vol_{\d B_{\rn^\tau}}(\d B_{\rn^\tau})}
\end{align*}
where $\Omega$ is the observable $\Phi^2$. Alternatively, if the analytic continuation cannot be computed explicitly, then we can use the fact that $E_0(z)=-im\omega^2\langle\Omega\rangle(z)$ is comprised of integrals that are well-posed for $\Re(z)\ll0$. Hence, $E_0(z)$ can be computed for a sufficiently large number of values $z$ with $\Re(z)\ll0$. The $z=0$ value can then be obtained by fitting the $z$-dependence (which is known up to some parameters) to the thus generated data and extrapolating the fit to $z=0$. In either case, we need to evaluate high-dimensional spherical integrals. 

\subsection{Solving the spherical integrals}
The numerical difficulty of computing
\begin{align*}
  E_0=&\frac{m\omega^2\tau}{2}\frac{\int_{\d B_{\rn^\tau}}\Omega\l(\sqrt\Sigma^{-1}\Psi\r)d\vol_{\d B_{\rn^\tau}}(\Psi)}{\vol_{\d B_{\rn^\tau}}(\d B_{\rn^\tau})}
\end{align*}
is the quadrature of the high-dimensional spherical integrals. We are using ``Sobol' points on the sphere $\d B_{\rn^\tau}$'' which can be generated from a ``standard'' Sobol' sequence~\cite{sobol} in three steps.
\begin{enumerate}
\item[Step 1.] Generate a $\tau$-dimensional Sobol' sequence (uniform distribution on $(0,1)^\tau$).
\item[Step 2.] Apply the standard normal ($\Np(0,1)$) percent point function $\mathrm{ppf}$ in each coordinate (standard normal distribution on $\rn^\tau$).
\item[Step 3.] Normalize each sample (uniform distribution on $\d B_{\rn^\tau}$).
\end{enumerate}
In order to make full use of the properties of the Sobol' sequence, we generate $2^n$ Sobol' points for some $n\in\nn$ and skip the first point since the first (non-randomized) Sobol' point is always $\l(\frac12,\ldots,\frac12\r)^T$ which gets mapped to the origin under $\mathrm{ppf}$ and thus cannot be normalized.

Let $\Sp$ be the set of generated Sobol' points. Then we can approximate $E_0$ via
\begin{align*}
  E_0\approx\frac{m\omega^2\tau}{2}\frac{1}{\#\Sp}\sum_{\Psi\in\Sp}\Omega\l(\sqrt\Sigma^{-1}\Psi\r).
\end{align*}
The convergence in terms of the number of (randomized) Sobol' samples is shown in \autoref{fig:sobol_error}. The standard deviations, which were extracted from $100$ simulations each, indicate an error scaling proportional to $\frac{1}{\sqrt{\#\Sp}}$. In other words, the $\zeta$-regularized lattice simulation of the harmonic oscillator with Lorentzian background is comparable to a Markov Chain Monte-Carlo simulation in Euclidean background in terms of numerical effort. 

\begin{figure}
  \includegraphics[width=\textwidth]{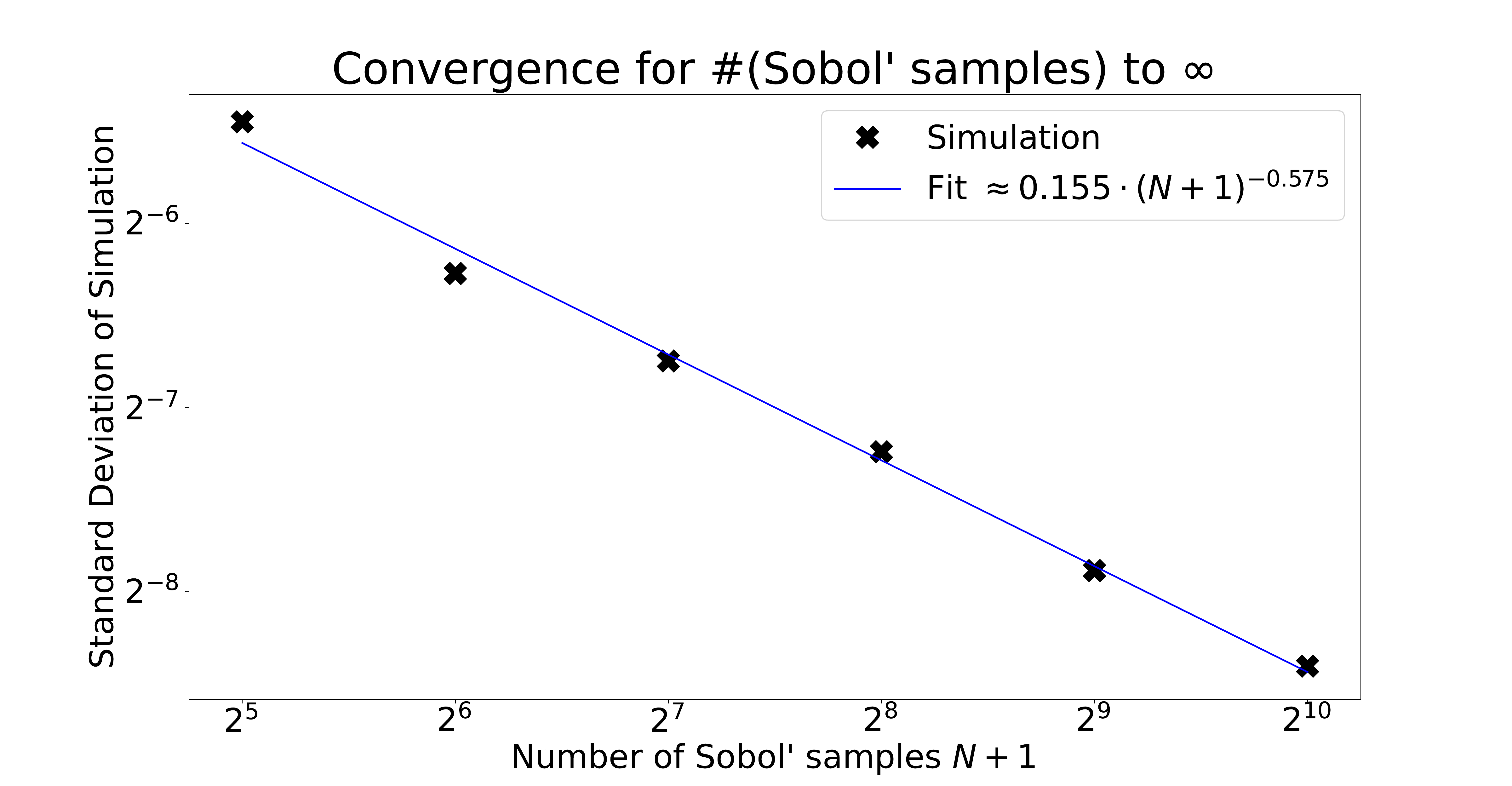}
  \caption{\label{fig:sobol_error}This figure shows the standard deviation of $\frac{1}{\#\Sp}\sum_{\Psi\in\Sp}\Omega\l(\sqrt\Sigma^{-1}\Psi\r)$ using $N=\#\Sp=2^n-1$ sample points. The standard deviation is estimated using $100$ simulations with randomized Sobol' points. Furthermore, a least square error fit $C (N+1)^\alpha$ to the simulation data is shown. The exponent $\alpha\approx-0.575$ indicates an error scaling $\mathrm{error}\propto\frac{1}{\sqrt{\#\Sp}}$. The simulation parameters are $m=\omega=\hbar=1$, $\tau=128$, and $a=0.01$.}
\end{figure}

Having a reasonable quadrature to estimate 
\begin{align*}
  E_0\approx\frac{m\omega^2\tau}{2}\frac{1}{\#\Sp}\sum_{\Psi\in\Sp}\Omega\l(\sqrt\Sigma^{-1}\Psi\r),
\end{align*}
we can now consider the behavior of $E_0(m,\omega,\hbar)$ as a function of the lattice spacing $a$ and the number of lattice points $\tau$. The physical volume is then given by $T=a\tau$. In order to extract the continuum limit, we are interested in the limit $a\searrow0$ at constant physical volume $T$ and the limit $T\nearrow\infty$. \autoref{fig:continuum_limit} (top) shows simulation results of $E_0(m=1,\omega=1,\hbar=1)$ with lattice spacing $a\in\{2^{-2},2^{-3},2^{-4},2^{-5}\}$ at physical volume $T=32$. Similarly, \autoref{fig:continuum_limit} (bottom) shows simulation results of \mbox{$E_0(m=1,\omega=1,\hbar=1)$} with lattice spacing $a=2^{-5}$ at physical volume $T\in\{4,8,16,32\}$. The error bars are the standard deviations extracted from $100$ simulations using $2^{10}-1$ randomized Sobol' points and the horizontal lines are the analytic continuum limit $E_0(m=1,\omega=1,\hbar=1)=\frac12$. In particular, we note that the results for $a=2^{-5}$ and $T\in\{16,32\}$ are already comparable to the continuum limit at the given level of accuracy. 

\begin{figure}[ht!]
  \begin{tabular}{c}
    \includegraphics[width=\textwidth]{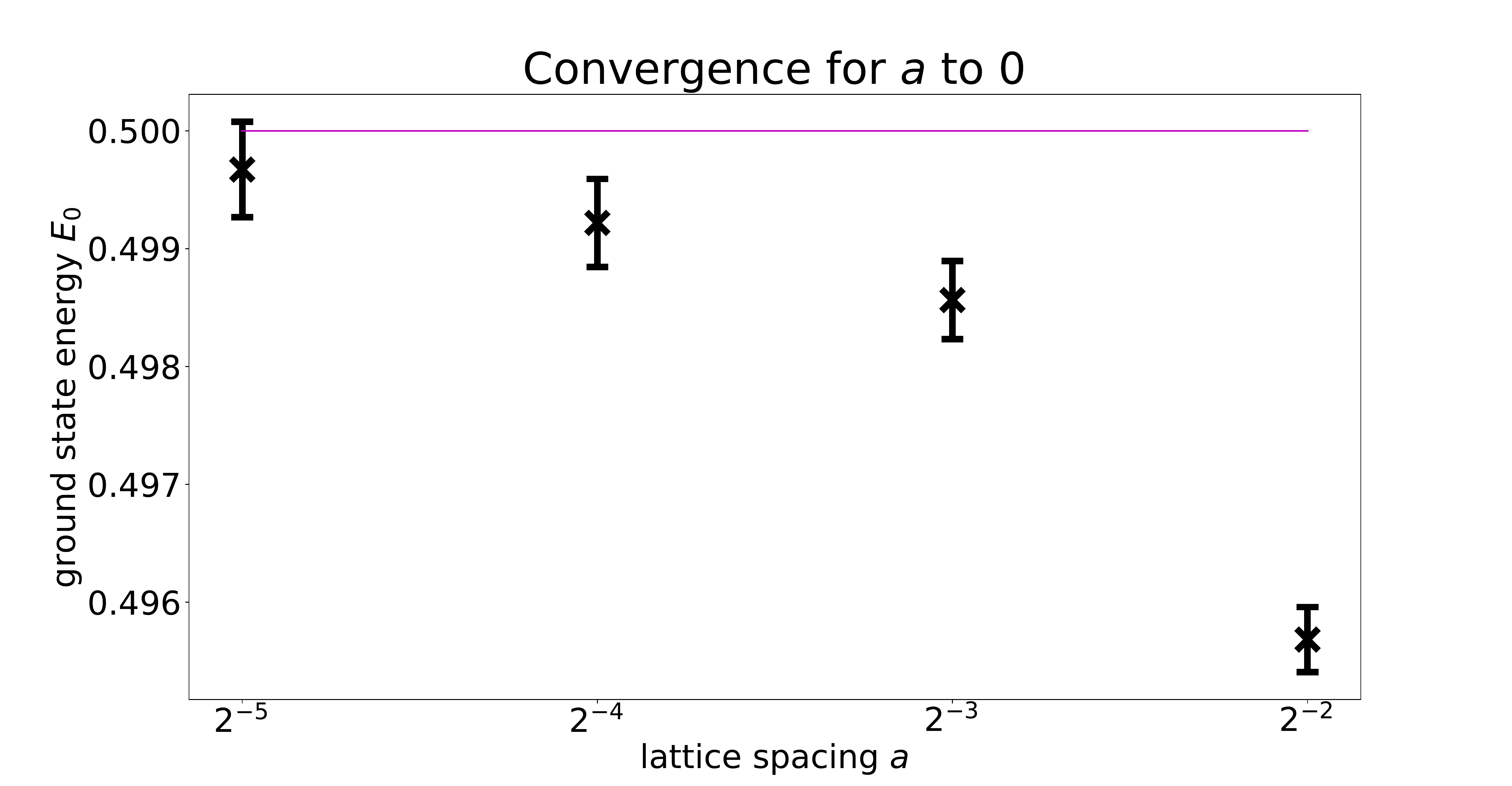}\\
    \includegraphics[width=\textwidth]{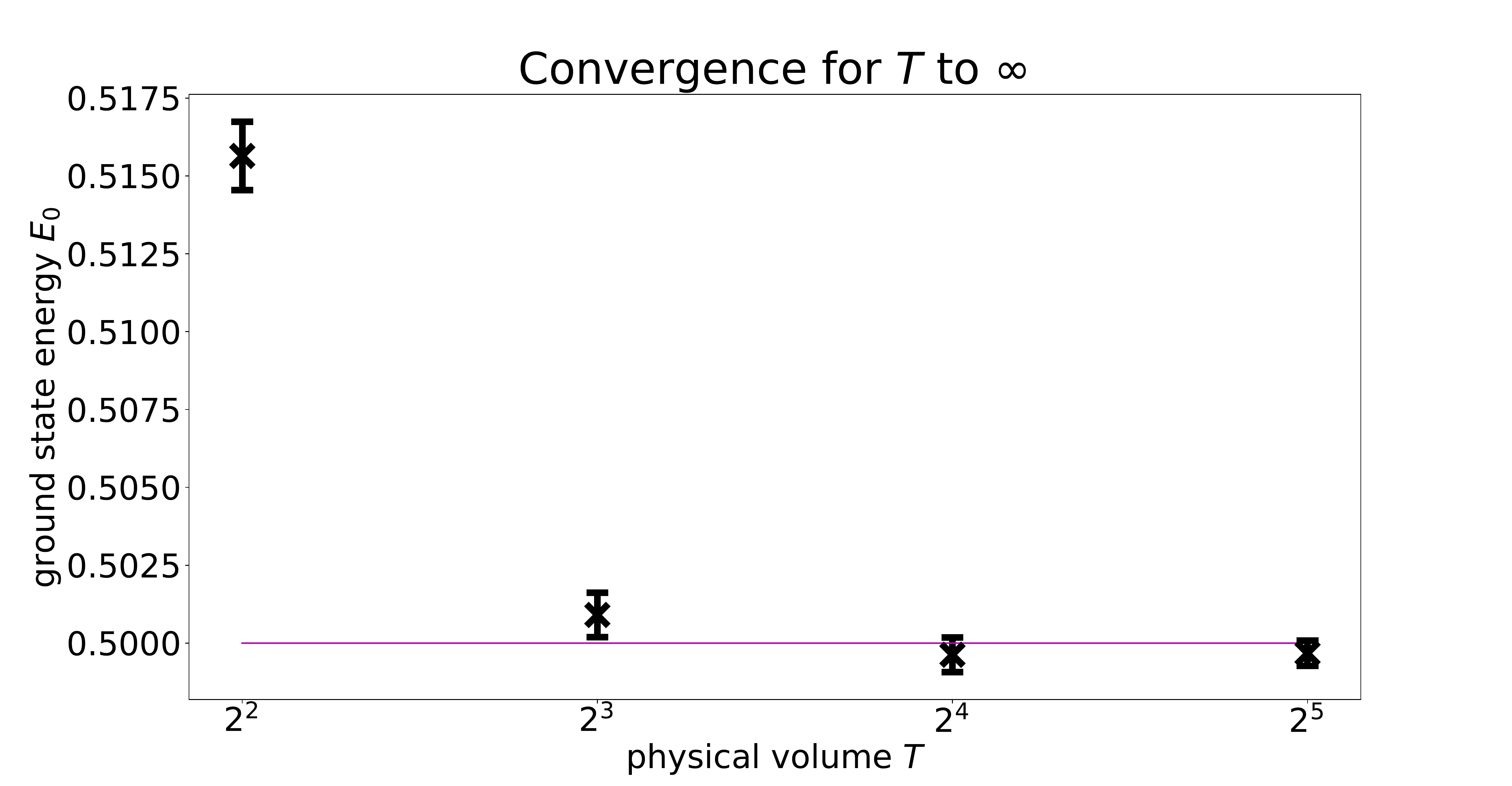}
  \end{tabular}
  \caption{\label{fig:continuum_limit}This figure shows physical simulations indicating the lattice spacing $a\to0$ limit at constant physical volume $T=32$ (top) and the physical volume $T\to\infty$ limit (bottom) at constant lattice spacing $a=2^{-5}$. The horizontal lines are the continuum ground state energy $E_0=\frac{\hbar\omega}{2}$. In the simulations, $m=\omega=\hbar=1$ was used and error bars denote standard deviations extracted from $100$ independent simulations. Each simulation uses $2^{10}-1$ randomized Sobol' points. It should be noted that the vertical axis range is within $1\%$ of $E_0$ at the top and within $4\%$ of $E_0$ at the bottom.}
\end{figure}

\subsection{Extrapolation from $\Re(z)\ll0$} In more complicated situations, it may not always be possible to compute the dependence on the gauge parameter $z$ explicitly, or computing the spherical integrals may be more difficult than computing the $\rn^\tau$ integrals 
\begin{align*}
  \langle\Omega\rangle(z)=\frac{\int_{\rn^\tau}e^{iS(\Phi)}\gf(z)(\Phi)\Omega(\Phi)d\Phi}{\int_{\rn^\tau}e^{iS(\Phi)}\gf(z)(\Phi)d\Phi}.
\end{align*}
In such cases, it may be advantageous to compute $\langle\Omega\rangle(z)$ at $\Re(z)\ll0$ where the integrals are well-defined and numerically well-posed. In order to guide the extrapolation to $z=0$, we need to account for the unknown $z$-dependence $\Delta(z)$. However, if the gauge is chosen well, then we can use the fact that the $z$-dependence comes from analytic continuations of Laplace transforms of poly-($\log$)-homogeneous distributions, for which the general form of these Laplace transforms in terms of the variable~$z$ are known. Thus, by combining those constituent terms, we obtain a parametric expression $\Delta(z,p)$ such that $\Delta(z,p^*)=\langle\Omega\rangle(z)$ is satisfied for some set of parameters~$p^*$. Once this general dependence $\Delta(z,p)$ is identified by analyzing the poly-($\log$)-homogeneity structure of the integrands, the parameters $p^*$ can be obtained from fitting $\Delta(z,p)$ against the numerical data $\langle \Omega\rangle(z)$. Finally we obtain $\langle \Omega\rangle(0)$ by extrapolation, i.e., by evaluating $\Delta(z,p^*)$ at $z=0$.

\begin{example*}
  Consider an action $S$ for which there exists a closed compact manifold~$M$ in $\rn^\tau$ such that $\fa \Phi\in\rn^\tau\ \ex! \phi\in M\ \ex! r\in\rn_{>0}:\ \Phi=r\phi$ and $S(\Phi)=c r^\sigma$ for some constant $c\in\rn$. Furthermore, let $\Omega(\Phi)=\sum_{\iota\in I}r^{\delta_\iota}\omega_\iota(\phi)$ be a polyhomogeneous observable and choose the gauge $\gf(z)(\Phi)=(1+r)^z$. Then,
  \begin{align*}
    \langle\Omega\rangle(z)=&\frac{\int_{\rn^\tau}e^{iS(\Phi)}\gf(z)(\Phi)\Omega(\Phi)d\Phi}{\int_{\rn^\tau}e^{iS(\Phi)}\gf(z)(\Phi)d\Phi}\\
    =&\frac{\sum_{\iota\in I}\int_{\rn_{>0}}e^{icr^\sigma}(1+r)^zr^{\delta_\iota+\tau-1}dr\int_M\omega_\iota(\phi)d\vol_M(\phi)}{\int_{\rn_{>0}}e^{icr^\sigma}(1+r)^zr^{\tau-1}dr\vol_M(M)}.
  \end{align*}
  Thus,
  \begin{align*}
    I_\delta(z):=&\int_{\rn_{>0}}e^{icr^\sigma}(1+r)^zr^\delta dr\\
    =&\int_0^1e^{icr^\sigma}(1+r)^zr^\delta dr+\int_{\rn_{>1}}e^{icr^\sigma}(1+r)^zr^\delta dr\\
    =&\sum_{k\in\nn_0}{z\choose k}\l(\int_0^1e^{icr^\sigma}r^{\delta+k} dr+\int_{\rn_{>1}}e^{icr^\sigma}r^{\delta+z-k} dr\r)\\
    =&\sum_{k\in\nn_0}{z\choose k}\frac{1}{\sigma}\l(\int_0^1e^{icr}r^{\frac{\delta+k+1-\sigma}{\sigma}} dr+\int_{\rn_{>1}}e^{icr}r^{\frac{\delta+z-k+1-\sigma}{\sigma}} dr\r)\\
    =&\sum_{k\in\nn_0}{z\choose k}\frac{1}{\sigma}\l(\frac{\gamma\l(\frac{\delta+k+1}{\sigma},-ic\r)}{(-ic)^{\frac{\delta+k+1}{\sigma}}}+\frac{\Gamma\l(\frac{\delta+z-k+1}{\sigma},-ic\r)}{(-ic)^{\frac{\delta+z-k+1}{\sigma}}}\r)
  \end{align*}
  implies
  \begin{align*}\tag{$***$}\label{eq:fit_Re(z)<<0}
    \langle\Omega\rangle(z)=&\sum_{\iota\in I}C_{\iota}\frac{\sum_{k\in\nn_0}{z\choose k}\frac{1}{\sigma}\l(\frac{\gamma\l(\frac{\delta_\iota+\tau+k}{\sigma},-ic\r)}{(-ic)^{\frac{\delta_\iota+\tau+k}{\sigma}}}+\frac{\Gamma\l(\frac{\delta_\iota+\tau+z-k}{\sigma},-ic\r)}{(-ic)^{\frac{\delta_\iota+\tau+z-k}{\sigma}}}\r)}{\sum_{k\in\nn_0}{z\choose k}\frac{1}{\sigma}\l(\frac{\gamma\l(\frac{\tau+k}{\sigma},-ic\r)}{(-ic)^{\frac{\tau+k}{\sigma}}}+\frac{\Gamma\l(\frac{\tau+z-k}{\sigma},-ic\r)}{(-ic)^{\frac{\tau+z-k}{\sigma}}}\r)}.
  \end{align*}
  In this expression, everything is known a priori with the exception of the $C_\iota$. In other words, we can fit \eref{eq:fit_Re(z)<<0} against simulation data to extract the $C_\iota$ and obtain
  \begin{align*}
    \langle\Omega\rangle(0)=&\sum_{\iota\in I}C_{\iota}\frac{\Gamma\l(\frac{\delta_\iota+\tau}{\sigma}\r)}{(-ic)^{\frac{\delta_\iota}{\sigma}}\Gamma\l(\frac{\tau}{\sigma}\r)}.
  \end{align*}
\end{example*}
\begin{figure}
  \begin{tabular}{c}
    \includegraphics[height=.25\textheight]{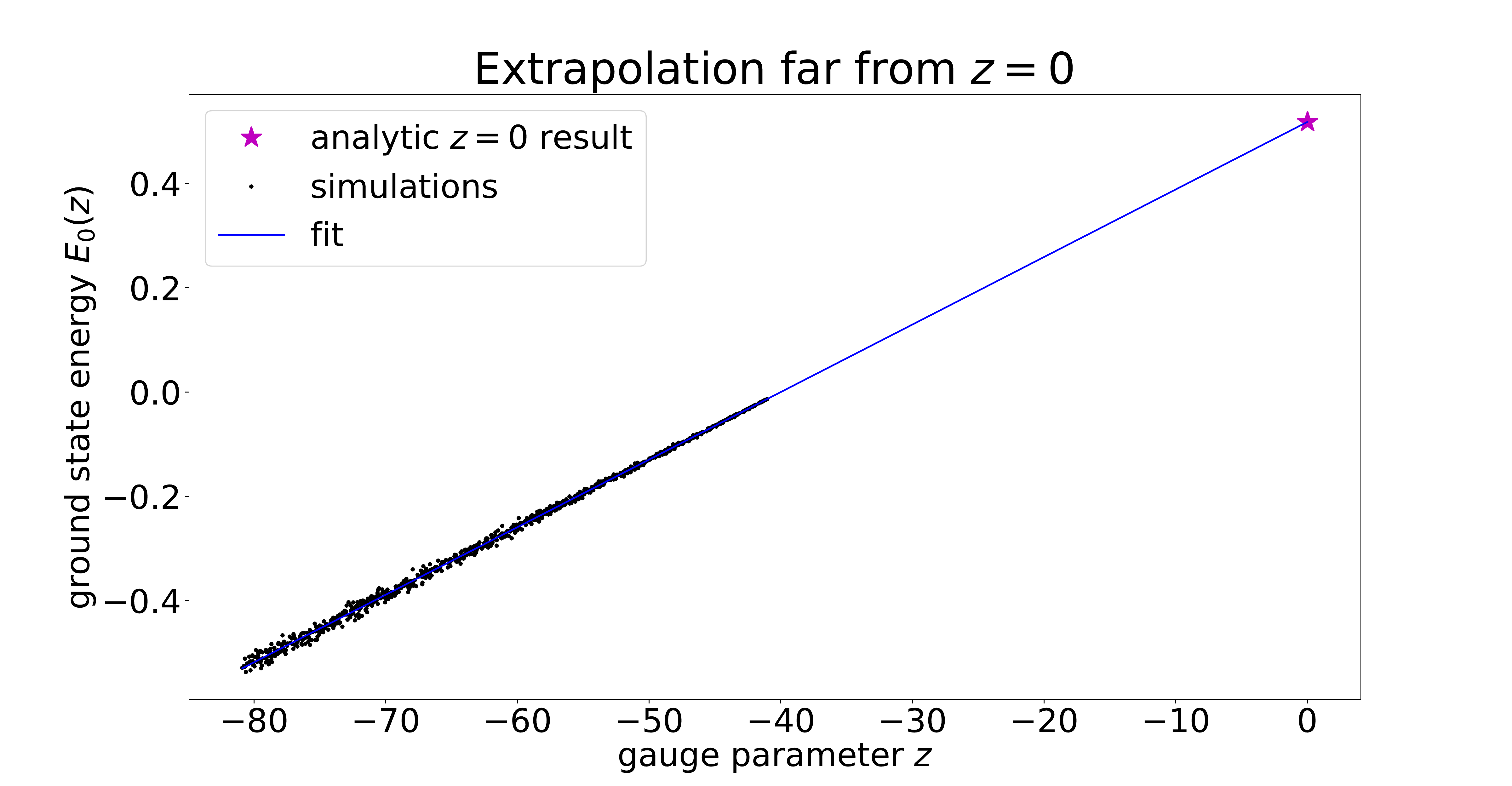}\\
    \includegraphics[height=.25\textheight]{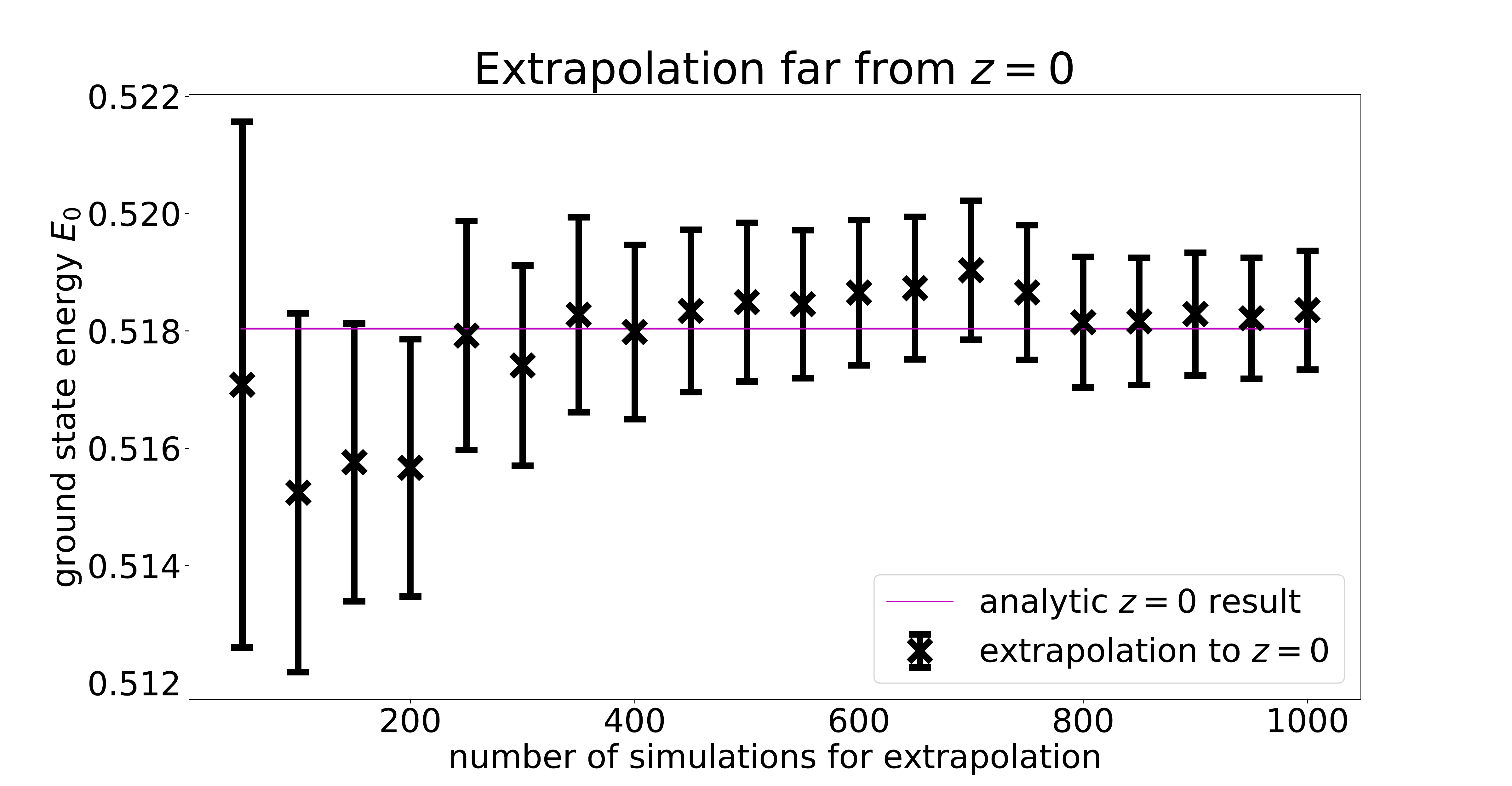}\\
    \includegraphics[height=.25\textheight]{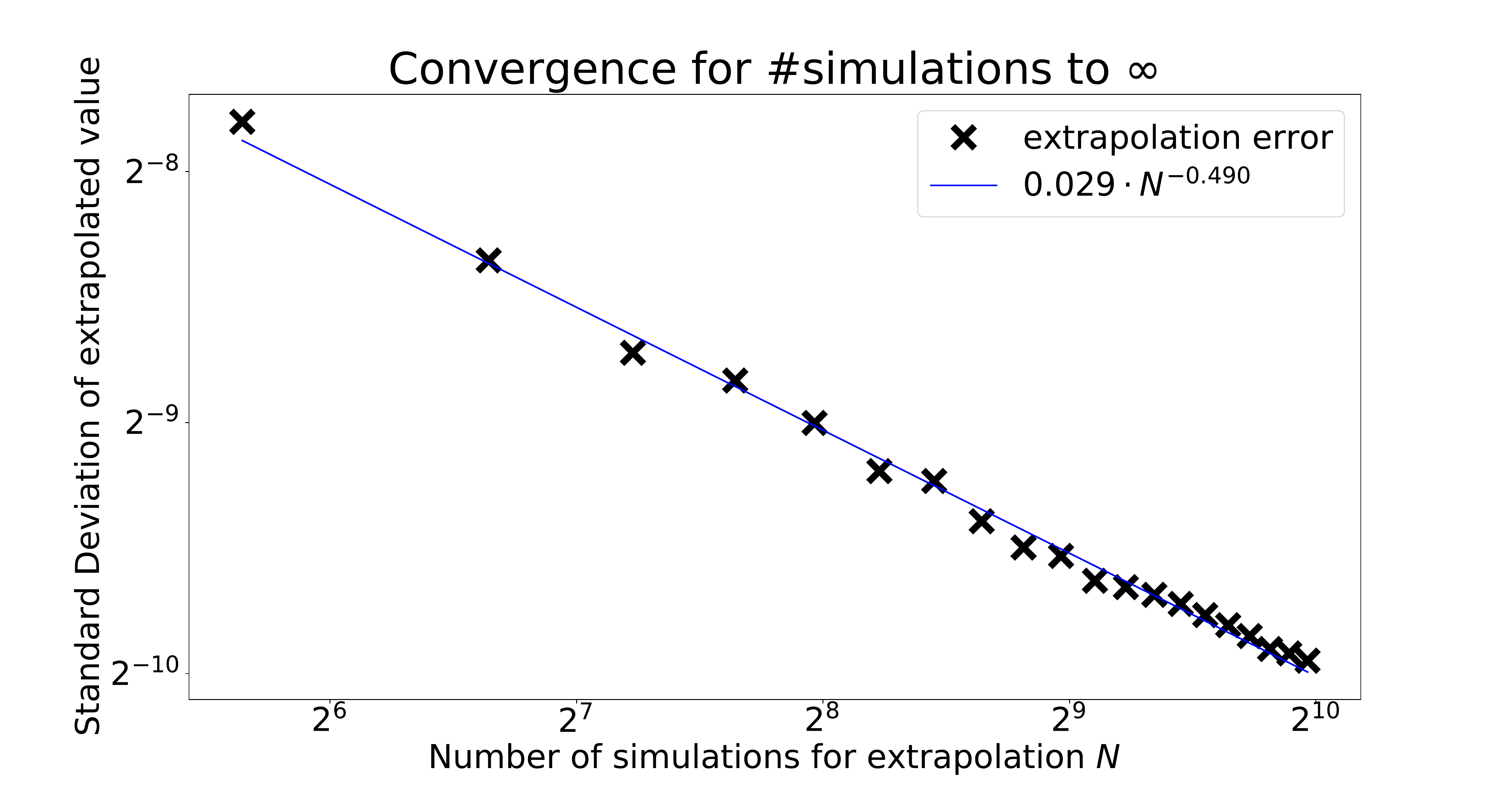}
  \end{tabular}
  \caption{\label{fig:extrapolation} This figure shows the continuum extrapolation of $1000$ simulations with randomly chosen $z\in[-2\tau-1,-\tau-1]$, $\tau=40$, $a=0.1$, and $m=\omega=\hbar=1$. The top plot shows the simulation data, the fit, and the $z=0$ value $E_0(a,\tau)$. The middle plot shows the extrapolated $z=0$ values with standard deviation as error using $50,100,\ldots,950,1000$ simulations and compares these to the $z=0$ value $E_0(a,\tau)$. The bottom plot shows the standard deviations of the extrapolated $z=0$ values as well as a least square error fit $C N^\alpha$ where $N$ is the number of simulations used for the $z=0$ extrapolation. The exponent $\alpha\approx-0.490$ indicates an error scaling $\mathrm{error}\propto\frac{1}{\sqrt{\#\mathrm{simulations}}}$.}
\end{figure}

With the choice of gauge $g(z)(r,\phi)=r^z$, the $z$-dependence $\Delta(z,p)$ is much simpler; namely affine linear $\Delta(z,p)=p_0+p_1z$. Using up to $1000$ simulations at randomly chosen $z\in[-2\tau-1,-\tau-1]$ with $\tau=40$, $a=0.1$, and $m=\omega=\hbar=1$, we can fit $\Delta(z,p)$ and obtain the $z=0$ extrapolation $p_0$. \autoref{fig:extrapolation} (top) shows the least square error fit of $\Delta(z,p)$ to the data of $1000$ simulations as well as the $z=0$ value $E_0(a,\tau)$. In \autoref{fig:extrapolation} (middle) the extrapolated value $p_0$ extracted from $50,100,\ldots,950,1000$ simulations including the standard deviation as error is shown. The comparison to the $z=0$ value $E_0(a,\tau)$, which we can compute analytically as in \autoref{sec:harmonic_oscillator_analytic}, shows compatibility of the extrapolated value $p_0$ with $E_0(a,\tau)$ already for a relatively small number of simulations. In \autoref{fig:extrapolation} (bottom), we compare the standard deviation of $p_0$ extracted from $50,100,\ldots,950,1000$ simulations to a least square error fit of the form $C\cdot(\#\text{simulations})^\alpha$. The fitted value of $\alpha\approx-0.490$ indicates an error scaling proportional to $\frac{1}{\sqrt{\#\text{simulations}}}$.

\section{Conclusion}
In this article, we have studied the $\zeta$-regularization of lattice field theories with Lorentzian background metrics. In particular, we have discussed the connection between lattice field theories with Euclidean backgrounds and Lorentzian backgrounds in \autoref{sec:zeta-reg-lattice}, in order to understand how the $\zeta$-regularization arises as a na\"ive means of regularizing lattice expectation values when it is not possible to analytically continue the theory to Euclidean time. In \autoref{sec:gauging-transfer-matrix}, we have then made the connection from the na\"ively $\zeta$-regularized lattice theory to the Hamiltonian formulation so far discussed in the literature~\cite{hartung-iwota,hartung-jmp,hartung-jansen,hartung-jansen-gauge-fields,jansen-hartung}. This allowed us to gain insights into the two main questions that need considering when extending the $\zeta$-regularization in the Hamiltonian formulation to the Lagrangian formulation commonly used to describe lattice field theories.
\begin{enumerate}
\item[Q1:] What kind of gauge families $\gf$ can be used to $\zeta$-regularize lattice field theories in a Lorentzian background?
\item[Q2:] Under what conditions is it possible to construct a $\zeta$-regularized Hamiltonian theory starting from a ``na\"ively'' $\zeta$-regularized lattice field theory? 
\end{enumerate}
Question Q1 is at least partially answered by \autoref{thm:compactification_limit} and \autoref{thm:analytic continuation}. In order to $\zeta$-regularize a lattice field theory in a Lorentzian background, $\gf(z)$ must be chosen in such a way, that the corresponding integrals are well-defined for $\Re(z)\ll0$. This is generally satisfied if $\gf(z)$ has an asymptotic behavior of $\abs{\gf(z)(\Phi)}\le c(1+\norm\Phi^2)^{\delta\Re(z)}$ or exhibits even faster convergence to zero as $\norm\Phi\to\infty$. Furthermore, for an increasing, exhaustive sequence of compacta $(K_n)_{n\in\nn}$, the compactified vacuum expectation values $\l(\frac{\langle 1_{K_n^T}\Omega\gf\rangle_L}{\langle 1_{K_n^T}\gf\rangle_L}\r)_{n\in\nn}$ need to be locally bounded in $C(D)$ where $D\sse\cn$ is open, connected, and contains $0$ as well as a half-space with $\Re(z)\ll0$. This local boundedness condition is likely to be the most intricate assumption to satisfy, as this is the likely point at which renormalization and modeling decisions can have a major impact.

It is important to note that the answer to Q1 merely implies the
existence of a $\zeta$-regularized lattice theory. Without further information, there is no reason to assume that this $\zeta$-regularized lattice theory is at all physical. In order to conclude physicality of the expectation values extracted from the $\zeta$-regularized lattice theory, we need to connect the Lagrangian formulation of the $\zeta$-regularized lattice theory to the Hamiltonian formulation studied to date~\cite{hartung-iwota,hartung-jmp,hartung-jansen,hartung-jansen-gauge-fields,jansen-hartung} which includes a proof of physicality. As such, answering Q2 allows us to make physical sense of the $\zeta$-regularized lattice vacuum expectation values.

Question Q2 is partially answered by \autoref{thm-gauged-transfer-matrix} and \autoref{thm-gauged-time-evolution}. In essence, we need to ensure that the resulting gauged transfer matrix or gauged time evolution operator is a suitable family of Fourier integral operators; namely a gauged family in the sense of~\cite{hartung-iwota,hartung-jmp,hartung-jansen,hartung-jansen-gauge-fields,jansen-hartung}. This can be ensured by choosing $\gf(z)$ or $\gf_\Delta(z)$ to be in a H\"ormander class $S^{\delta\Re(z)}$ or (poly-)homogeneous. While this restricts the class of gauge families $\gf$ sufficient to construct the $\zeta$-regularized lattice theory as per question Q1, computationally this is not a bad class of functions as many analytic results are known and, thus, allowing for computational simplifications in numerical simulations of $\zeta$-regularized lattice theories.

In \autoref{sec:classical-limit}, we then considered the classical limit of $\zeta$-regularized lattice theories and observed the remarkable fact that the classical limit is independent of the value for $z$. If we consider the impact of the gauge on a classical action, this can be rationalized since the gauge adds an imaginary component to a classical action. Thus, if we expect for only the real-part of the gauged action to be retained in the classical limit, then \autoref{thm:classical-limit} confirms this point of view.

Finally, we illustrated the theory of $\zeta$-regularized lattice theories developed in this article at the example of the harmonic oscillator. First, we treated it analytically in \autoref{sec:harmonic_oscillator_analytic}; second, we considered it numerically in \autoref{sec:harmonic_oscillator_numeric}. Both the analytic and numeric studies of the harmonic oscillator focus on the ground state energy.

In \autoref{sec:harmonic_oscillator_analytic}, we computed the ground state energy of the harmonic oscillator on the lattice five times. Initially, we computed the ground state energy by using a Wick rotation, i.e., using the canonical approach of analytic continuation to imaginary (Euclidean) time. This served as a comparison for the four choices of gauge families that follow, as well as to show that the lattice theory has the correct continuum limit. The four gauge families that were used to compute the $\zeta$-regularized harmonic oscillator on the lattice with Minkowski background highlight the four cases arising from two binary choices we can make. On one hand, we have the choice of gauging local in time vs. global in time, that is, gauging the transfer matrix or the entire time evolution operator. On the other hand, we can choose between absolutely integrable gauge families and distributionally regularizing gauge families, where the former is more directly applicable to numerical simulations but the latter allows for analytic shortcuts via the Laplace transform. Of course, all four gauge families reproduced the correct ground state energy, as is to be expected since we have gauge-independence.

The numerical study of the $\zeta$-regularized harmonic oscillator on the lattice with Minkowski background in \autoref{sec:harmonic_oscillator_numeric} focused on the numerical obstacles arising in classical simulations of $\zeta$-regularized lattice theories as opposed to quantum simulations which are a special case of the quantum simulations discussed in~\cite{hartung-jansen,jansen-hartung}. In particular, we noted that distributionally regularizing gauge families allow for numerical shortcuts via Laplace transform and can make for simpler functional dependence on the complex parameter $z$, which needs to be fitted against simulation data for extrapolation if the extrapolation cannot be done analytically to begin with. However, distributionally regularizing gauge families require numerical quadratures over high-dimensional manifolds (a sphere for the harmonic oscillator). Absolutely integrable gauge families, on the other hand, ``only'' require quadratures in high-dimensional vector spaces $\rn^\tau$. As more powerful integration techniques exist for integrating over $\rn^\tau$ than codimension~$1$ submanifolds of $\rn^\tau$, the absolutely integrable gauge families might be preferable from a numerical point of view. This question may be exacerbated further when considering gauge fields in lattice field theories. Unfortunately, the cost of choosing absolutely integrable gauge families is a more complicated dependence on the complex parameter $z$ which may make extrapolation to $z=0$ more difficult.

\end{document}